\documentclass[acmsmall]{ec21acm}

\AtBeginDocument{%
	\providecommand\BibTeX{{%
			\normalfont B\kern-0.5em{\scshape i\kern-0.25em b}\kern-0.8em\TeX}}}

\usepackage{amsmath}
\usepackage{amsthm}
\usepackage{subfig}
\usepackage{color}
\usepackage[english]{babel}
\usepackage{graphicx}
\usepackage{grffile}
\usepackage{xspace}
\usepackage{wrapfig,epsfig}
\usepackage{epstopdf}
\usepackage{url}
\usepackage{color}
\usepackage{epstopdf}
\usepackage{algpseudocode}
\usepackage[T1]{fontenc}
\usepackage{bbm}
\usepackage{comment}
\usepackage{dsfont}
\usepackage{thm-restate}
\usepackage{bm}
\usepackage{tcolorbox}

\usepackage{amsmath}
\usepackage{tikz}
\usepackage{thm-restate}
\usepackage{enumerate}

\def\isArxiv{1}

\DeclareMathOperator*{\E}{{\mathbb{E}}}

\DeclareMathOperator{\poly}{poly}
\DeclareMathOperator{\Supp}{Supp}

\newcommand{\eps}{\epsilon}
\newcommand{\bs}{\mathbf{s}}
\newcommand{\bx}{\mathbf{x}}
\newcommand{\be}{\mathbf{e}}

\newcommand{\tx}{\tilde{x}}
\newcommand{\twd}{\mathrm{twd}}

\newcommand{\bDelta}{\mathbf{\Delta}}
\newcommand{\g}{\mathcal{G}}
\newcommand{\R}{\mathbb{R}}
\newcommand{\T}{\mathcal{T}}
\newcommand{\mS}{\mathcal{S}}
\newcommand{\gc}{\mathrm{GCIRCUIT}}
\newcommand{\threeSAT}{\text{3SAT}\xspace}
\newcommand{\nm}[1]{\left|\, #1\, \right|\xspace}

\newtheorem{theorem}{Theorem}[section]
\newtheorem{lemma}[theorem]{Lemma}
\newtheorem{definition}[theorem]{Definition}

\copyrightyear{2021} 
\acmYear{2021} 
\setcopyright{acmcopyright}\acmConference[EC '21]{Proceedings of the 22nd ACM Conference on Economics and Computation}{July 18--23, 2021}{Budapest, Hungary}
\acmBooktitle{Proceedings of the 22nd ACM Conference on Economics and Computation (EC '21), July 18--23, 2021, Budapest, Hungary}
\acmPrice{15.00}
\acmDOI{10.1145/3465456.3467616}
\acmISBN{978-1-4503-8554-1/21/07}

\acmSubmissionID{ec0068f}

\begin{document}
	
	\title{Public Goods Games in Directed Networks}
	
	\author{Christos Papadimitriou}
	\email{christos@columbia.edu }
	\author{Binghui Peng}
	\email{bp2601@columbia.edu}
	\affiliation{%
		\institution{Columbia University}
	}

	\renewcommand{\shortauthors}{Papadimitriou and Peng}
	
\begin{abstract}
	Public goods games in undirected networks are generally known to have pure Nash equilibria, which are easy to find.  In contrast, we prove that, in directed networks, a broad range of public goods games have intractable equilibrium problems: The existence of pure Nash equilibria is NP-hard to decide, and mixed Nash equilibria are PPAD-hard to find.  We define general utility public goods games, and prove a complexity dichotomy result for finding pure equilibria, and a PPAD-completeness proof for mixed Nash equilibria.  Even in the divisible goods variant of the problem, where existence is easy to prove, finding the equilibrium is PPAD-complete.  Finally, when the treewidth of the directed network is appropriately bounded, we prove that polynomial-time algorithms are possible.
\end{abstract}

	

	\maketitle
	

	\section{Introduction}
\label{intro}

A public good is a resource which, once produced, is available to all (non-excludability), and can be enjoyed collectively by many agents (non-rivalry\footnote{Non-rivalry was called {\em collective consumption} by Paul Samuelson, who initiated the study of the subject \cite{samuelson1954pure}.}). Scientific knowledge \cite{stiglitz1999knowledge}, open-source software, vaccination for an infectious disease, volunteer work, information resources, and clean environment are fine examples of public goods. Since public goods can be produced at a cost and contribute to the utility of others, they enable a variety of strategic behaviors such as free-riding.  Game theoretic formulations of public goods have been extensively studied by economists --- see~\cite{bg86} for a classical framework for the public goods problem within which a unique Nash equilibrium exists.  

Networks are perfect arenas for public goods games~~\cite{bramoulle2007public}.  Networks model the fact that a particular public good, such as a piece of software or protection due to the immunization of an individual, may not be accessible by all, but only by the neighbors of the node where it is produced.  A node's utility then is an nondecreasing function of the goods in the neighborhood, minus the cost of the goods produced by the node.  Almost all of the literature deals with the homogeneous case, where all nodes have the same two strategies (produce the common goods at a cost, or not) and the same utility function (see \cite{yu2020computing} for an exception); in fact, the nondecreasing functions max and sum are typically considered.  In this paper, we assume that all nodes have the same utility (even though they have different circumstances due to network connectivity, and hence the game is not symmetric unless the graph is), and we consider very general utility functions.  There is extensive work on public goods in undirected networks (see the related work subsection), and the rough consensus seems to be that, in just about all variants of the problem (again, with the exception of \cite{yu2020computing}), pure Nash equilibria exist --- typically corresponding to independent or dominating sets of the graph --- and are easy to find.  

Undirected graphs have limitations as models of utility transfer.  The ability to enjoy the public goods produced by others is not necessarily symmetric --- for example, clean air in a neighboring city is of no use if that city is downwind; if my house is on a co-worker's way to work, then the good of carpooling to work produced by my co-worker benefits me, but not vice-versa; while social networks are often directed: Twitter, Instagram, Flickr, and others.  There has been some work on public goods in directed networks, e.g.~\cite{lopez2013public}, where sufficient conditions for the existence of equilibria are developed.  The general impression one gets from the literature is that the matter of equilibria in the directed case is more subtle.



{\em This paper is a comprehensive exploration of the complexity of the equilibrium problem in public goods games on directed graphs.} 

The simplest and most widely studied variant of the problem is the indivisible case with the max utility: the decision a node faces is whether or not to produce the good; and a node does not need to produce the good if one or more of its predecessors have it.  It turns out to be quite intricate.  It is easy to see that a pure equilibrium may not exist (consider a directed odd cycle), and it turns out that it is NP-complete to decide if a pure equilibrium does exist.  We give a simple reduction to that effect (Theorem~\ref{thm:np-hard}).

We then generalize this NP-completeness result to a full {\em complexity dichotomy} of nondecreasing utility functions.  We identify three families of utility functions that can be solved in polynomial time:  The {\em flat functions}, {\em the steep functions}, and {\em the alternating functions.}  The first two have trivial equlibria where all nodes abstain or all nodes produce the good, respectively.  In the case of alternating functions, finding a pure Nash equilibrium is shown to be equivalent to solving a system of equations in $\bf F_2$.  The main part of the proof entails showing that all other functions make the equilibrium problem NP-hard (Theorem \ref{thm:np-hardgen}).

Since pure equilibria in these games are fraught with non-existence and NP-completeness, can we find in polynomial time a mixed Nash equilibrium (guaranteed to exist by Nash's theorem)?    We prove (Theorem \ref{thm:ppad-mix}) that this problem is PPAD-complete, even in the simplest case of the max utility; this is perhaps the most technically demanding proof in this paper.
We reduce from the generalized circuit problem, proved to be PPAD-hard in~\cite{rubinstein2018inapproximability, chen2009settling}. The reduction requires several new ideas, including the definition of a new kind of intermediate game --- in addition to several that already exist in this literature --- which we call {\em the threshold game}, and we believe is of interest in its own right.  
Finally, when the goods are divisible, the {\em max} case of the problem (the utility is the maximum of the neighbors minus the good's cost) is one where it can be seen through a simple averaging argument that a mixed equilibrium exists; we show that this too is PPAD-complete to find, this time by a reduction from mixed Nash equilibria in two-player win-lose games~\cite{chen2007approximation,abbott2005complexity}.


All of our complexity results hold for sparse networks, with indegrees and outdegrees at most three.  But how about networks that are tree-like {\em in the sense of graph minors}~\cite{robertson1986graph}?   We show that, when the (underlying undirected) network has {\em bounded treewidth}, essentially all versions of the Nash equilibrium problem of network public goods games can be solved, or at least approximated arbitrarily close, in polynomial time. Our algorithm and techniques are inspired by \cite{daskalakis2006computing} and \cite{thomas2015pure}, but several substantial adaptations and innovations are needed.

\paragraph{Our contributions.}   In summary, our main contributions are these:
\begin{itemize}

\item  Sweeping intractability results for the equilibrium problem of public goods games in directed networks, including a novel PPAD-completeness proof through threshold games, an intriguing analysis of polynomial special cases for the pure equilibrium problem culminating in a precise P/NP-complete dichotomy, and even a very different PPAD-completeness proof for divisible goods.

\item The formulation of public goods games in networks with a general objective function --- beyond the two functions treated in the literature, max and sum --- leading to a surprisingly rich and diverse family of problems (Section 3).  In the discussion section we point out that the complexity of such classes of equilibrium problems is open even for undirected networks.

\item An approximation algorithm when the treewidth is $O({\log n/ \log\log n})$, through the development of new and enhanced techniques for approximating equilibrium problems in graphical games with small treewidth. 
\end{itemize}

\subsection{Related work}
\label{sec:related-work}

Bramoull\'e and Kranton~\cite{bramoulle2007public} initiated the study of public goods in a network.  They consider a type of pure Nash equilibrium called {\em specialized equilibrium}, and prove that such equilibria 
are {\em stable} under small perturbations, {\em universal} (always exist), and in fact {\em computable} by a natural distributed algorithm, since they correspond to {maximal independent sets} of the graph; see~\cite{dall2011optimal,boncinelli2012stochastic,lopez2013public, feldman2013pricing, bramoulle2014strategic, allouch2015private,shin2017you,elliott2019network,yu2020computing,kempe2020inducing} for follow-up works.
Bramoull\'e et. al.~\cite{bramoulle2014strategic} extended the theory to imperfectly substitutable public goods, and proved the existence of a unique Nash equilibrium, assuming that the graph's lowest eigenvalue is sufficiently small. 
Allouch~\cite{allouch2015private} differentiates private provision from public provision, and again characterizes the existence and uniqueness of a Nash equilibrium through the lowest eigenvalue of the graph.
Public goods games were first generalized to directed graphs in~\cite{lopez2013public}, who provide sufficient conditions for pure Nash equilibria to exist. The only complexity result regarding such public goods games we are aware of is~\cite{yu2020computing}: finding a pure Nash equilibrium of a discrete version of the public goods game, albeit in the far more general case of {\em heterogeneous} agents, is NP-hard. 
We refer interested readers to the surveys~\cite{jackson2015games, galeotti2010network,bramoullegames} for a general coverage of this area.


Our work uses certain ideas from {\em graphical games}~\cite{kearns2001graphical}. 
It is NP-hard to find a pure Nash equilibrium~\cite{gottlob2005pure} and PPAD-hard to compute, even approximately, a mixed Nash equilibrium~\cite{daskalakis2009complexity, rubinstein2018inapproximability, chen2009settling} of a general graphical game with maximum degree $3$.
However, the problem is tractable in several settings~\cite{daskalakis2006computing, thomas2015pure, daskalakis2015approximate}. Daskalakis and Papadimitriou~\cite{daskalakis2006computing} developed a polynomial-time approximation scheme (PTAS) for computing an $\eps$-approximate Nash equilibrium  when the game has bounded strategy size, the network has bounded neighborhood size and $O(\log n)$ treewidth. Thomas and Leeuwen~\cite{thomas2015pure} provided an algorithm that computes a pure Nash equilibrium in $\poly(s^{w}, |M|)$, where $s$ is the strategy size, $w$ the {\em treewidth} of the graph and $|M|$ the size of the payoff matrix. We use similar ideas in our main algorithmic result for computing Nash equilibria in public goods problems for networks of bounded treewidth, but we have to address the problem that, in the present case, the parameter $|M|$ of this algorithm is exponential.

The PPAD complexity class was introduced by Papadimitriou~\cite{papadimitriou1994complexity} to capture one particular genre of total search functions, encompassing the notion of equilibrium. 
The PPAD-completeness of Nash equilibria was established in~\cite{daskalakis2009complexity,chen2009settling} and extended recently in~\cite{rubinstein2018inapproximability,rubinstein2016settling}. 
Over the past decades, a broad range of problems have been proved to be PPAD-hard, including equilibrium computation~\cite{daskalakis2009complexity1,abbott2005complexity, chen2007approximation,chen2015complexity}, 
market equilibrium~\cite{chen2009spending, chen2009settling-market, vazirani2011market, chen2011optimal,chen2013complexity}, equilibrium in auction~\cite{chen2021complexity,filos2021complexity}, fair allocation~\cite{othman2016complexity}, min-max optimization~\cite{daskalakis2021complexity} and problems in financial networks~\cite{schuldenzucker2017finding}.  





	\section{Model}
\label{sec:model}
A {\em public goods game} is a game with $n$ players, defined through a directed graph $G(V, E)$ without loops, where $V=\{1, \ldots, n\}$ is the set of players. 
We use $N(i)$ to denote the {\em neighborhood} of $i$, namely incoming neighbors of agent $i$, i.e., $N(i) = \{i\} \cup \{j | (j,i)\in E\}$. 
We assume common game theoretic terms and notation, such as strategy, strategy profile, pure Nash equilibrium and (mixed) Nash equilibrium.  
If $\bs = (s_1, \ldots, s_n)$ is a strategy profile, we use $s_{-i}$ to denote actions adopted by all agents {\em except} $i$. 

As is almost always done with public goods games, we assume that {\em all players have the same strategy space and the same utility function}. 
In the indivisible good (discrete) case, the strategy space of all players is $S=\{0,1\}$, while in the divisible (continuous) case $S=[0,\infty)$. 
To define the {\em utility function} $U_i$ of a player $i$, we start that defining the {\em price} or {\em cost} $p$ of producing the good $s_i$, common to all players.  In the indivisible case, it is a single real $p(s_i) = p >0$.  In the divisible case it is a function $p:\R_+\rightarrow \R_+$.

Once $p$ has been fixed, the common utility function of agent $i$ for the strategy profile $\bs$ is $U_i(\bs) = X_i(\bs) - p(s_i)$, where $X_i$ is a {\em symmetric social composition function of the strategies played by the players in $N(i)$.}
Since players may have different indegrees, and thus different sizes of neighborhood, we assume for uniformity that the common social composition function $X$ is a {\em symmetric} function from $S^n$ to the reals, where the strategies of players not in $N(i)$ are all set to zero --- a value that does not affect $X$.  
The composition functions studied by the vast majority of the literature is the {\em max} (or {\em best shot}, or {\em or}\/) function in the indivisible case, picking the maximum of the neighborhood's $0-1$ choices, while in the divisible case the composition functions {\em max} and {\em sum} is used.

In indivisible good games with {\em max} composition, in the literature it is always assumed that $p\neq 1$, because otherwise $p=1$ creates ties between contributing and free-riding.  For more general indivisible good games and social composition functions $X$, we shall also avoid ties between contributing and free-riding.  This can be achieved by assuming that all values of $X$ are rational (not a significant loss of generality), while $p$ is also rational but with a large denominator (technically, larger than the square of the largest denominator used in the values of $X$).  This completes the definition of the common general utility function $U$, and thus of the game.

\if{false}
\begin{align}
\label{eq:indivisible-utility}
u_i = \left\{\begin{matrix}
0 &\sum_{j\in N(i)}s_j = 0 \\
U-p & s_i = 1\\
U & \sum_{j\in N(i)}s_j \geq 1 \textbf{ and }s_i = 0.
\end{matrix}
\right.
\end{align}
That is, we assume that access to multiple goods does not increase the benefit (unit demand) and the value of a good does not decrease if the good is shared by multiple agents (non-rivalry). 
The second assumption are rather standard in the literature~\cite{dall2011optimal,shin2017you,boncinelli2012stochastic}, while the first is called in the literature the {\em best-shot} game of~\cite{hirshleifer1983weakest}, here generalized to directed graphs.

\subsection{Divisible Good} 
\label{sec:model-divisible-good}
Divisible public goods are scalars, such as air quality and volunteer work.
The strategy profile $S_i$ of each agent $i$ is $[0, +\infty)$ and we assume a cost $p$ per unit of the good. 
Utility depends on the {\em social composition functions} $X(\cdot)$, where $X: \R^{n}\rightarrow \R^{n}$ maps a strategy profile $\bs$ to a vector $X(\bs)$ that indicates the amount of goods available for each agent.
The usual choice of social composition in the literature is the summation function~\cite{bramoulle2007public}, i.e., $X(\bs)[i] = \sum_{j \in N(i)} s_j, \, \forall i$; we also consider the best-shot rule (max) and the weakest link rule (min)~\cite{hirshleifer1983weakest}:
\begin{itemize}
	\item $X(\bs)[i] = \max_{j \in N(i)} s_j\,\, \forall i$ \quad  Best-shot.
	\item $X(\bs)[i] = \min_{j \in N(i)} s_j\,\, \forall i$ \quad Weakest-link.
\end{itemize}
	
The utility of an agent $i$ is defined as 
\begin{align}
\label{eq:divisible-util}
u_i = U\left(X(\bs)[i]\right) - p\cdot s_i,
\end{align}
where $U(\cdot)$ denotes the valuation function of the good. 
We assume that the valuation function $U$ is non-decreasing and concave, i.e., the provision of public good brings positive benefit and it has diminishing returns.  It is also useful to assume that $U$ is differentiable, at least almost everywhere.
Notice that, if we set the valuation function to be $U(x) = U\cdot \min\{x,  1\}$ and restrict the strategy space of each agent to $\{0,1,2,\ldots\}$, 
then both the summation and best shot rule collapse to the indivisible case of the previous subsection.

\fi
We are interested in the standard concepts of pure and mixed Nash equilibrium.  A strategy profile $\bs = (s_1, \ldots, s_n)$ of the public good game is a {\em (pure) Nash equilibrium}, 
if no agent can derive better utility by changing their own strategy,
\begin{align*}
u_{i}(s_i, s_{-i}) \geq u_{i}(s_{i}', s_{-i}) \quad \forall i\in V,\,\forall s_i, s_{i}'\in S_i.
\end{align*} 

In a {\em mixed Nash equilibrium} $\bDelta = (\Delta_1, \ldots, \Delta_n)$, each agent $i$ plays a distribution $\Delta_i$ over its strategy set $S_i$, and satisfies

\begin{align}
\label{eq:nash-def1}
\E_{s_i \sim \Delta_i, s_{-i}\sim \Delta_{-i}}\left[u_{i}(s_i, s_{-i})\right] \geq \E_{s_{-i}\sim \Delta_{-i}}\left[u_{i}(s_{i}', s_{-i})\right] \quad \forall i\in V,\,\forall s_{i}'\in S_i.
\end{align} 

Define $\Supp(\Delta_i)$ to be the support of the distribution $\Delta_i$, i.e., $\Supp(\Delta_i)= \{s_i | s_i \in S_i, \Delta_i(s_i) > 0\}$. Then the definition in~\eqref{eq:nash-def1} is equivalent to 
\begin{align*}
 \forall i\in V,\,\forall s_i \in \Supp(\Delta_i), s_{i}' \in S_i: \E_{s_{-i}\sim \Delta_{-i}}\left[u_{i}(s_i, s_{-i})\right] \geq \E_{s_{-i}\sim \Delta_{-i}}\left[u_{i}(s_{i}', s_{-i})\right].
\end{align*}
An {\em $\eps$-approximately well supported Nash equilibrium} ($\eps$-Nash) is then defined as
\begin{align}
\label{eq:nash-def3}
\forall i\in V,\,\forall s_i \in \Supp(\Delta_i), s_{i}' \in S_i: \E_{s_{-i}\sim \Delta_{-i}}\left[u_{i}(s_i, s_{-i})\right]  \geq \E_{s_{-i}\sim \Delta_{-i}}\left[u_{i}(s_{i}', s_{-i})\right] - \eps.
\end{align}

	\section{Pure Nash Equilibria: A Dichotomy}
\label{sec:indivisible}
In this section we characterize the complexity of finding pure equilibria, focusing first on the best shot (max, or) function.  In contrast to undirected networks, where every {\em maximal} independent set corresponds to a pure Nash equilibrium, pure Nash equilibria may not exist in directed graphs (see Figure~\ref{fig:pure-nash}).  We show in this section that determining whether a pure Nash equilibrium exists is NP-complete, and then generalize this to a sweeping complexity dichotomy result, characterizing precisely --- modulo the P$\neq$NP conjecture --- the kinds of utility functions that have tractable Nash equilibrium problems. 

\label{sec:np-hard}
	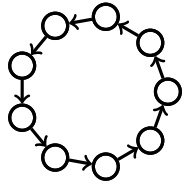
\begin{figure}[!ht]
		\centering
			\begin{tikzpicture}[scale=0.5]
			\foreach \x in {0,40,...,320} {
				\node  (\x) at (\x:2) [draw, circle, very thick] {};
				\node  (\x+1) at (\x + 40:2) [draw, circle, very thick] {};
				\draw[->, very thick] (\x) -- (\x+1);
			}
			\end{tikzpicture}
		\caption{Odd cycles have no pure Nash equilibrium.}
		\label{fig:pure-nash}		
	\end{figure}

\begin{theorem}
	\label{thm:np-hard}
Deciding whether a pure Nash equilibrium exists in an indivisible public good game with the max social composition function is NP-complete.
\end{theorem}
\begin{proof}
In an equilibrium profile $\bs = (s_1, \ldots, s_n)$, for each agent $i$, we have $s_i = 1$ if $\sum_{j\in N_i}s_j = 0$, and $s_i = 0$ otherwise. 
In another words, an agent would purchase the good if, and only if, none of its predecessors possesses the good. 
The reduction is from \threeSAT, and employs the following two gadgets (see Figure.~\ref{pic:variable-gadgets} and Figure.~\ref{pic:AND-gadgets}).

	\begin{figure}[!h]
		\centering
		\begin{minipage}[t]{.5\textwidth}
			\centering
			\begin{tikzpicture}
			\node (0) at (-1, 0) [circle, draw, very thick] {};
			\node (1) at (0, 0) [circle, draw, very thick] {};
			\node (2) at (1, 0) [circle, draw, very thick] {};
			\node (3) at (2, 0) [circle, draw, very thick] {};
			
			\draw[<->,  very thick] (0) -- (1);
			\draw[->,  very thick] (1) -- (2);
			
			\node[text width=1cm] at (-0.6, 0.4) {$x_i$};
			\node[text width=1cm] at (0.4, 0.4) {$\bar{x}_i$};
			\node[text width=1cm] at (1.4, 0.4) {$x_i$};
			\node[text width=1cm] at (1.8, -0.02) {$\cdots$};
			\node[text width=1cm] at (2.4, 0.4) {$\bar{x}_i$};

			\node[text width=1cm] at (-1,-1) {};
			\node[text width=1cm] at (-1,1) {};
			\end{tikzpicture}
			\caption{Variable gadget.}
			\label{pic:variable-gadgets}
		\end{minipage}%
		\begin{minipage}[t]{.5\textwidth}
			\centering
			\begin{tikzpicture}
			\node (0) at (-1, 1) [circle, draw, very thick] {0};
			\node (1) at (-1, 0) [circle, draw, very thick] {1};
			\node (2) at (-1, -1) [circle, draw, very thick] {2};
			\node (3) at (0, 0) [circle, draw, very thick] {3};
			\node (4) at (1, 0) [circle, draw, very thick] {4};
			
			\node (5) at (2, 0) [circle, draw, very thick] {5};
			\node (6) at (2.85, 0.5) [circle, draw, very thick] {6};
			\node (7) at (2.85, -0.5) [circle, draw, very thick] {7};

			\draw[->,  very thick] (0) -- (3);
			\draw[->,  very thick] (1) -- (3);			
			\draw[->,  very thick] (2) -- (3);
			\draw[->,  very thick] (3) -- (4);	
			\draw[->,  very thick] (4) -- (5);			
			\draw[->,  very thick] (5) -- (6);			
			\draw[->,  very thick] (6) -- (7);			
			\draw[->,  very thick] (7) -- (5);
			
			\node[text width=1cm, scale=1] at (-1.3, 1) {$x_1$};
			\node[text width=1cm, scale=1] at (-1.3, 0) {$\bar{x}_2$};
			\node[text width=1cm, scale=1] at (-1.3, -1) {$x_3$};
			
			\node[text width=2cm, scale=0.9] at (1.2, 0.5) {$ x_1\!\vee\! \bar{x}_2 \!\vee x_3$};

			\end{tikzpicture}
		\caption{Clause gadget.}
		\label{pic:AND-gadgets}
		\end{minipage}
	\end{figure}
	
\paragraph{Variable gadget}
	For each variable $x_i$, we construct a directed path with $2k_i$ nodes,  where $k_i$ is the number of times $x_i$ appears in the \threeSAT instance.  
The path is directed, with the exception that there is a bi-directional edge between the first two nodes.  The bi-directional edge forces the choice (exactly one of the first two nodes has the good), and the rest of the path propagates it (either all odd nodes have the good and all even nodes do not, or the other way around).
		
\paragraph{Clause gadget}
	The clause gadget consists of two parts. The left part is an OR gadget, in that node 4 must equal the disjunction of nodes $0, 1,$ and $2$.  To see this, suppose that none of these three nodes has the good; then node 3 must have it, and so node $4$ does not.  And if one or more of nodes $0, 1, 2$ has the good, then 3 does not have the good, and thus 4 must have it.

The right part forces the clause to be true --- that is, in any equilibrium profile, agent $4$ must play strategy 1 and buys the good. 
This is because if $4$ does not provide the good, then $5,6,7$ is an isolated odd circle, which cannot exist in a pure Nash equilibrium.
On the other hand, players $4$ and $6$ buying the good, and players $5$ and $7$ not buying it, is a pure Nash equilibrium of the four rightmost nodes. In summary, the clause gadget ensures that at least one of the nodes $0,1,2$ buys the good.
	
Putting things together, given any \threeSAT instance we can construct a public good game by composing variable gadgets and clause gadgets in the obvious way, so that the pure Nash equilibria of the public good game are in one to one correspondence with the satisfiable solutions of the \threeSAT instance, concluding the proof.
\end{proof}

We want to generalize this result to {\em any} social composition function $X$, and so we start with the question: {\em For which composition functions is the pure Nash equilibrium problem polynomial-time solvable?} Consider a symmetric, non-decreasing function $X:\{0,1\}^n\mapsto \R_+$ without loss of generality with $X(0^n)=0$.  Because of symmetry, we can treat $X$ as a function from $\bf N$ to $R_+$, since its value depends on $\sum_{i=1}^n s_i$; we shall use the same symbol for this form of $X$\footnote{That is, we assume that $X$ has values for all integers, not limited to the size of the network; this is obviously a harmless convention.}, and recall that $X(0)=0$.  Since $X$ is monotone, it can be also thought as a sequence of nonnegative {\em steps}.  Call $X$ {\em flat} if $X(1)\leq p$; that is, the first {step} of $X$ does not provide sufficient incentive to produce the good.  Obviously, all flat functions have the all-zero pure Nash equilibrium, and so the problem is trivial.   Call now $X$ {\em steep} if for all $k\geq 0, X(k+1) \geq X(k)+p$; that is, all steps are at least $p$.  Then all nodes have an incentive to produce the good no matter what anybody else is doing, and so the all-ones solution is a pure Nash equilibrium, and again the problem is trivial.
We have shown:

\begin{lemma}
The pure Nash equilibrium problem is in P if the utility function is flat.  Ditto for steep functions.
\end{lemma}

Are there any other tractable cases?  It turns out, that there is one more: Call $X$ {\em alternating} if for all $k\geq 0$, $X(k+1) < X(k)+p$ if $k$ is odd, and $X(k+1) > X(k)+p$ if $k$ is even.  

\begin{lemma}
The pure Nash equilibrium problem is in P if the utility function is alternating.
\end{lemma}

\begin{proof}
Let $s = (s_1, \ldots, s_n)$ be the equilibrium profile with $s_i \in \{0, 1\}$. Based on the definition of alternating utility function, we have that for any $i \in [n]$
\[
s_i = \left\{
\begin{matrix}
1 & \sum_{(j, i) \in E}s_i =0 (\bmod 2)\\
0 & \sum_{(j, i) \in E}s_i =1 (\bmod 2).
\end{matrix}
\right.
\]
That is, a player chooses to produce when there is an even number of neighboring players who produce the good. Hence, the equilibrium problem reduces to the solution of a linear system of equations in ${\bf F}_2$ with one $0-1$ variable $s_i$ per player, with one equation for each player $i$:
$$s_i + \sum_{(j,i)\in E} s_j = 1 (\bmod 2).$$
This can be solved in polynomial time with Gaussian elimination, say.
\end{proof}
\smallskip

We next establish that, unless P = NP,  {\em these are the only tractable cases:}
\begin{theorem}
	\label{thm:np-hardgen}
If the utility function does not belong in these three classes: (1) flat; (2) steep; or (3) alternating, then the pure Nash equilibrium problem is NP-complete.
\end{theorem}
\begin{proof}
We use the variable gadgets and the clause gadgets in the proof of Theorem~\ref{thm:np-hard} (see Figure~\ref{pic:variable-gadgets} and Figure\ref{pic:AND-gadgets}), but we reduce from several different NP-hard problems. First,  observe that when $X$ is not flat, steep, or alternating, there must be a $k\geq 0$ such that $X(k+1) > X(k) +p$, $X(k+2) < X(k+1) +p$. We start by assuming that $k=0$, that is, $X(1) > p$ and $X(2) < X(1) + p$. We divide the proof into four cases.

{ \bf Case 1 \ \ } Suppose $X(3) < X(2) + p$, $X(4) < X(3) + p$, then it is easy to check that the construction of Theorem~\ref{thm:np-hard} works. Indeed, a function satisfying $X(1) >p, X(2) < X(1) +p, X(3) < X(2) + p, X(4) < X(3) + p$ is, for the purposes of the network constructed in the proof of the previous theorem, equivalent to the max function.

{ \bf Case 2 \ \ } Suppose $X(3) < X(2) + p$, $X(4) > X(3) + p$. We can still use the network constructed in the proof of Theorem~\ref{thm:np-hard}. The difference is that we reduce from {\sc Not-all-equal SAT}, since the clause gadget is satisfied if and only if one or two literals are true.

{ \bf Case 3 \ \ } Suppose $X(3) > X(2) + p$, $X(4) > X(3) + p$. Again, we consider the network constructed in the proof of Theorem~\ref{thm:np-hard}. We can check that the clause gadget is satifiable if and only if exactly one of the literal is true. The NP-hardness then comes from {\sc One-in-3SAT}.

{ \bf Case 4 \ \ } Suppose $X(3) > X(2) + p$, $X(4) < X(3) + p$. Since we assume $X$ is not alternating, there exists $t \geq 1$ satisfying $X(2t+1) > X(2t) + p$, $X(2t+2) < X(2t+1) + p$, and $X(2t+3), X(2t+4)$ does {\em not} obey $X(2t+3) > X(2t+2) + p$, $X(2t+4) < X(2t+3) + p$. We create $2t$ new players who have no incoming edges and directed edges to all other nodes. These $2t$ players will provide the good at equilibrium, and thus the remaining players start the game with $2t$ copies of the good already. The game for the original player is then changed to the function $\tilde{X}(j) = X(j - 2t)$, and NP-completeness follows from cases (1--3).


Finally, suppose that $k>0$.  Add $k$ new players who have no incoming edges, and directed edges to all other nodes.  At equilibrium, these nodes will provide the good, and so the remaining players will start the game with $k$ copies of the good already provided.  Therefore, the game for the remaining players will be as if $X(j)$ was changed to $X(j-k)$, that is to say, to a function covered by the previous paragraph. 

We only need to be cautious about one exception: $X(1) > p, \ldots, X(k + 1) > X(k) + p, X(k+2) < X(k+1) + p, X(k+3) > X(k+2) + p, X(k+4) < X(k+3) + p$, since $X$ could be alternating after $X(k)$.
We still create $k$ new players and direct them to all other agents, {\em except} for node $3$ in every clause gadget, for which we only connect $k - 1$ players to it. It is easy to check that our argument in Theorem~\ref{thm:np-hard} works, with one modification: the clause gadget is satisfiable if and only if exactly two of the literals are true. This, again, is NP-complete, as we can reduce from {\sc One-in-3SAT}.
\end{proof}

\section{PPAD-hardness of Mixed Nash Equilibria}
\label{sec:ppad-hard-mix}
We next examine mixed Nash equilibria of indivisible public goods games.
In a mixed Nash equilibrium, agents randomize over the two actions and choose to buy the public good with some probability. 
We denote by $s_i$ the probability that agent $i$ purchases the good.
 Also, by $x = y\pm \eps$ we mean that $y - \eps\leq x \leq y+\eps$.
 For ease of presentation, we assume $U = 1$ and $p < 1$ throughout the proof.
 The following result is the main technical contribution of this paper.

\begin{theorem}
	\label{thm:ppad-mix}
	There exists some constant $\eps > 0$, such that it is PPAD-hard to find an $\eps$-Nash of the indivisible public goods game.
\end{theorem}

We start with a high level overview of the proof.  
We reduce from the $\eps$-$\gc$ problem (see Definition~\ref{def:generelized-circuit}), which is shown to be PPAD-hard for sufficiently small constant $\eps > 0$ by Rubinstein~\cite{rubinstein2018inapproximability}.
Our reduction consists of two steps.
We first introduce an intermediate game, called the {\em threshold game} (see Definition~\ref{def:threshold-game}), where each individual's strategy depends solely on the {\em summation} of its neighbors' strategies.
The threshold game exhibits rich algorithmic and complexity structure, which we believe could be of independent interest. We show a correspondence between equilibrium profiles of threshold games and those of public goods games (Lemma~\ref{lem:threshold-to-public}); hence it suffices to demonstrate PPAD-hardness of finding an $\eps$-approximate equilibrium of threshold games.
The threshold game is semi-anonymous, and thus we can only modify the local structure of the graph in order to construct all 9 types of gates $\{G_{\xi}, G_{\times \xi}, G_{=}, G_{+}, G_{-}, G_{<}, G_{\wedge}, G_{\vee}, G_{\neg}\}$. 
We start by constructing an elementary gadget $G_{\frac{1}{2}-}$ (see Figure~\ref{fig:element-gadget}), and use it as a building block to gradually construct most of the gates.
We restrict the players' equilibrium strategies to be in $[0, \frac{1}{2} + \eps]$ in the {\em arithmetic} (non-logic) gates.
The logic gates $(G_{\wedge}, G_{\vee}, G_{\neg})$ are special: in order to prove PPAD-hardness for constant $\eps$, we need them to be error-resilient, in that they do not amplify errors of the input.
We achieve this by restricting players' equilibrium strategy to be $\{0, 1\}$ (instead of $\{0, \frac{1}{2}\}$) when doing logic operations, and construct {\em transformer gadgets} $G_{\frac{1}{2}, 1}, G_{1, \frac{1}{2}}$ to map between the domains $\{0, \frac{1}{2}\}$ and $\{0, 1\}$.

\subsection{Equivalence between public goods games and threshold games}
\label{sec:equivalence}
We first introduce the {\em threshold game}.
\begin{definition}(Threshold game) 
	\label{def:threshold-game}
	A threshold game $\g(V, E, t)$ is defined on a directed graph $G = (V, E)$, with a threshold $t$ $(0 < t < 1)$. The vertices of the graph represent players with strategy space $[0, 1]$.  A strategy profile $\bx = (x_1, \ldots, x_n) \in [0, 1]^{n}$ is an equilibrium if it satisfies
	\begin{align}
	x_{i} = 
	\left\{\begin{matrix}
	0 & \sum_{j \in N_i}x_j > t\\
	1 & \sum_{j \in N_i}x_i < t \\
	\text{arbitrary }&\sum_{j \in N_i}x_i = t
	\end{matrix}\right..
	\end{align}
Note that $x_i$ can be an arbitrary number in $[0, 1]$ if $\sum_{j \in N_i}x_i = t$.
\end{definition}

We define the $\eps$-approximate equilibrium in a threshold game as follows:
\begin{definition}($\eps$-approximate equilibrium of threshold game) 
	\label{def:approx-threshold-game}
Let $\eps > 0$ be a constant satisfying $\eps < t < 1-\eps$.  An $\eps$-approximate equilibrium $\bx = (x_1, \ldots, x_n)\in [0, 1]^n$ of a threshold game $\g(V, E, t)$ satisfies
	\begin{align}
	x_{i} = 
	\left\{\begin{matrix}
	0 \pm \eps & \sum_{j \in N_i}x_j > t + \eps\\
	1 \pm \eps & \sum_{j \in N_i}x_i < t - \eps \\
	\text{arbitrary }&\sum_{j \in N_i}x_i \in [t - \eps, t + \eps]
	\end{matrix}\right..
	\end{align}
\end{definition}


We next establish the equivalence between threshold games and public good games. The proof can be found in Appendix~\ref{sec:omit-proof}
\begin{restatable}{lemma}{lemReduction}
\label{lem:threshold-to-public}
There is a polynomial time reduction between the threshold game and the public good game. 
Specifically,
(1) given any threshold game $\g(V, E, t)$ with $0 < t < 1$, we can construct a public good game and map any $\eps$-Nash of the public goods game to an $8\eps$-approximate equilibrium of threshold game $\g(V, E, t)$, for $\eps < \min\{0.1,\frac{t}{8},\frac{1-t}{8}\}$; 
(2) given any public good game with $U = 1, 0 < p < 1$, we can construct a threshold game $\g(V, E, t)$ and map any $\eps$-approximate equilibrium of threshold game to an $c_p\eps$-Nash of public goods game, where $c_p = -4p\log p$ is a constant depending only on $p$.
\end{restatable}

\subsection{Reducing generalized circuits to threshold games} 
\label{sec:reduce-circuit}
Next, we give the definition of generalized circuits.

\begin{definition}(Generalized circuit~\cite{chen2009settling}) 
	\label{def:generelized-circuit}
	A generalized circuit is a tuple $(V, \T )$, where $V$ is a set of nodes and $\T$ is a collection of gates. Every gate $T\in \T$ is a 5-tuple $T = (G, v_1, v_2, v, \alpha)$, where $G \in \{G_{\xi}, G_{\times \xi}, G_{=}, G_{+}, G_{-}, G_{<}, G_{\wedge}, G_{\vee}, G_{\neg}\}$ is the type of the gate; $v_1, v_2 \in V \cup \{nil\}$ are the input nodes, $\alpha \in \R  \cup \{nil\}$ is a real parameter and $v$ is the output node.	
	
	The collection $\T$ of gates must satisfy the following important property. For every two gates $T, T' \in \T$, $T = (G, v_1, v_2, v, \alpha)$ and $T' = (G', v_{1}', v_{2}', v', \alpha')$, we must have $v \neq v'$.
\end{definition}

The $\eps$-$\gc$ is the problem of finding an $\eps$-approximate assignment for the generalized circuit.
Notice that we replace $G_{\xi}$, $G_{\times \xi}$ with $G_{\frac{1}{2}}$, $G_{\times\frac{1}{2}}$ for ease of proof.
\begin{definition}
	\label{def:generelized-circuit1}
	Given a generalized circuit $\mS = (V, \T)$, we say an assignment $\bx: V \rightarrow [0, 1]$ $\eps$-approximately satisfies $\mS$, if it satisfies the constraints shown in table~\ref{table:constraint}. 
	\begin{table}[!htbp]
	    \renewcommand\arraystretch{1.25}
		\centering
		\begin{tabular}{|c|c|}
			\hline
			Gate & Constraint\\
			\hline 
			\hline
			$G_{\frac{1}{2}}(v)$ & $\bx[v] = \frac{1}{2} \pm \eps$ \\
			\hline
			$G_{\times \frac{1}{2}}(\nm{v_1} v)$ & $\bx[v] = \frac{1}{2} \cdot \bx[v_1]  \pm \eps$ \\
			\hline
			$G_{=}( \nm{v_1} v)$ & $\bx[v] =  \bx[v_1]  \pm \eps$ \\
			\hline
			$G_{+}( \nm{v_1, v_2} v)$ & $\bx[v] =  \min\{\bx[v_1] + \bx[v_2], \frac{1}{2}\}  \pm \eps$ \\
			\hline
			$G_{-}(\nm{ v_1, v_2} v)$ & $\bx[v] =  \max\{\bx[v_1] - \bx[v_2], 0\}  \pm \eps$ \\
			\hline
			$G_{<}( \nm{v_1, v_2} v)$ & $\bx[v] = \left\{\begin{array}{cc}\frac{1}{2} \pm \eps &\bx[v_1] < \bx[v_2] - \eps\\ 0 \pm \eps &\bx[v_1] > \bx[v_2] + \eps \end{array}\right.$\\
			\hline
			$G_{\wedge}( \nm{v_1, v_2} v)$ & $\bx[v] = \left\{\begin{array}{cc}\frac{1}{2} \pm \eps &\bx[v_1] = \frac{1}{2}\pm \eps  \wedge \bx[v_2] = \frac{1}{2}\pm \eps \\ 0 \pm \eps &\bx[v_1] = 0 \pm \eps  \vee \bx[v_2] = 0\pm \eps \end{array}\right.$\\
			\hline
			$G_{\vee}( \nm{v_1, v_2} v)$ & $\bx[v] = \left\{\begin{array}{cc}1 \pm \eps &\bx[v_1] = \frac{1}{2} \pm \eps  \vee \bx[v_2] = \frac{1}{2}\pm \eps \\ 0 \pm \eps &\bx[v_1] = 0 \pm \eps  \wedge \bx[v_2] = 0\pm \eps \end{array}\right.$\\
			\hline
			$G_{\neg}( \nm{v_1}  v)$ & $\bx[v] = \left\{\begin{array}{cc}\frac{1}{2} \pm \eps &\bx[v_1] = 0 \pm \eps\\  0 \pm \eps & \bx[v_1] = \frac{1}{2}\pm \eps \end{array}\right.$\\
			\hline
		\end{tabular}
		\caption{}
		\label{table:constraint}
	\end{table}
\end{definition}

To complete the proof we must reduce $\eps$-$\gc$ to computing an $\eps$-approximate equilibrium of the threshold game. 

\begin{theorem}
\label{thm:threshold-hard}
It is PPAD-hard to find an $\eps$-approximate equilibrium of the threshold game, for some constant $\eps > 0$.
\end{theorem}

The proof of this result can be found in Appendix~\ref{sec:omit-proof}.
It is based on an elementary game gadget $G_{\frac{1}{2}-}(\nm{v_1, v_2} v)$, where $v_1, v_2 \in V\cup \{nil\}$ are input players, $v\in V$ is the output player. 
The gadget consists of a directed cycle, where $v_1, v_2$ have directed edges to the auxiliary player $v_a$, $v_a$ has a directed edge to the auxiliary player $v_b$, $v_b$ points to the output player $v$ and there is a directed edge from $v$ to $v_a$. The output player $v$ could have many outgoing edges, but it only has one incoming edge from the internal node $v_b$. 
The intention is that, at equilibrium, $\bx[v] = \max\{\frac{1}{2} - \bx[v_1] - \bx[v_2], 0\}$. 
The core of the proof, which we omit here, entails the construction and deployment of gates $G_{=}$, $G_{+}$, $G_{-}$, whose design is based on the elementary gadget described above. 

Combining Theorem~\ref{thm:threshold-hard} and Lemma~\ref{lem:threshold-to-public}, the proof of Theorem~\ref{thm:ppad-mix} is complete.

Theorem~\ref{thm:ppad-mix} only covers the max utility function.  We {\em conjecture} that it holds for all utility functions except for the polynomial cases discussed in Theorem \ref{thm:np-hardgen}.  There are several interesting challenges in extending the proof to this direction.
	\section{Divisible Goods}
\label{sec:divisible}
For divisible public goods games in directed graphs, we study the three most studied utility functions, and completely characterize the equilibrium problem: 
\begin{itemize}
\item For the summation utility (the utility of each node is the sum of the amounts of goods provided by its predecessors), a pure Nash equilibrium always exists, but it is PPAD-complete to find one.

\item For the best-shot function (max), a pure Nash equilibrium may not exist and it is PPAD-hard to find a mixed Nash equilibrium; the reasons are quite similar to the indivisible case. 
 
\item Finally, for the weakest-link function (min), it turns out that there are always multiple trivial pure Nash equilibria, as no player has the incentive to supply any amount of the good.
\end{itemize}

\subsection{Summation}
When the utility function is the summation, 
Bramoulle et. al.~\cite{bramoulle2007public}       
prove that there is always a pure Nash equilibrium\footnote{They actually prove it for  undirected networks, but their proof generalizes easily to the directed case, as the best response function for each agent $i$ is still continuous in $s_{-i}$, and existence of a pure Nash equilibrium follows from Brouwer's fix point theorem.  In fact, when the valuation function is strictly concave, it can be shown that there are {\em only} pure Nash equilibria ~\cite{bramoulle2007public}:  It is always better to replace the mixed strategy with its mean value, but also, replacing by the mean value does not create pure Neq.}.
Here, we prove it is PPAD-hard to find one, when the network is directed. 
In fact, we prove a slightly stronger result:
call a strategy profile $\bs = (s_1, \ldots, s_n)$ an {\em $\eps$-approximate pure Nash equilibrium}, if $s_i = b_i(s_{-i}) \pm \eps$, where $b_i(\cdot)$ is the best response of agent $i$.

\begin{theorem}
\label{thm:ppad-pure}
It is PPAD-hard to find an $\eps$-approximate pure Nash equilibrium of public goods games, for $\eps = 1 / \poly(n)$. 
\end{theorem}
We reduce from the mixed Nash equilibrium problem in two-player win-lose games. 
A two-player game $(R, C)$ is win-lose if $R, C \in \{0, 1\}^{n \times n}$.
It is known~\cite{chen2007approximation,abbott2005complexity} that finding an $\eps$-Nash of two-player win-lose game is PPAD-hard for $\eps = 1/\poly(n)$. 
Given an instance $(R, C)$ of a two player win-lose game, we construct a divisible public goods game on a directed network, such that we can map any $\eps$-approximate pure Nash equilibrium of the public goods game to a $\poly(n)\cdot\eps$-Nash of two-player win-lose game $(R, C)$. 
In order to do so, we first symmetrize the win-lose game (Lemma~\ref{lem:symmetric}), and then reduce it to the public goods game (Lemma~\ref{lem:pure-ppad}). 

For convenience, we assume that $R, C \in\{-1, 0\}^{n\times n}$ and that there is no weakly-dominated strategy for both row and column players. 
Moreover, we assume every column (row) of $R(C)$ contains at least one $0$ entry --- otherwise, there is a trivial pure Nash equilibrium.
Define a symmetric game $(A, B)$ as follows:
\begin{align*}
A = \left(\begin{matrix}
-\mathbf{1} & R\\
C^{T} & -\mathbf{1}\\
\end{matrix}
\right)
\text{ and }
B = \left(\begin{matrix}
-\mathbf{1} & C\\
R^{T} & -\mathbf{1}\\
\end{matrix}
\right),
\end{align*}
where $-\mathbf{1}$ denotes an $n\times n$ all -1 matrix. We notice that $A, B \in \{-1,0\}^{2n\times 2n}$ and $A = B^{T}$. The above symmetrization is standard in the literature~\cite{lemke1964equilibrium}, and it is known that for any symmetric Nash equilibrium $(x, y)$ of $(A, B)$, $(x/|x|, y/|y|)$ is a Nash equilibrium for $(R, C)$. The following lemma states that approximation is preserved:
\begin{lemma}
\label{lem:symmetric}
Suppose $(x, y)$ is a symmetric $\eps$-Nash equilibrium of game $(A, B)$.  Then $(\tilde{x}, \tilde{y}) = (x/|x|, y/|y|)$ is a 4n$\eps$-Nash of the win-lose game $(R, C)$.
\end{lemma}
\begin{proof}
We first prove that $|x|, |y| \geq \frac{1}{4n}$. Notice that
\begin{align*}
A\left(\begin{matrix}x\\y\end{matrix}\right) = 
\left(\begin{matrix}
-\mathbf{1} & R\\
C^{T} & -\mathbf{1}\\
\end{matrix}
\right)
\left(\begin{matrix}x\\y\end{matrix}\right)=
\left(\begin{matrix}-|x|\be + Ry\\-|y|\be + C^{T}x\end{matrix}\right),
\end{align*}
where $\be = \{1, \ldots, 1\}^{T}$ denotes the n-dimension all 1 vector. 
The proof is by contradiction. 
Suppose $|x| < \frac{1}{4n}$; it then follows that
\begin{align*}
\max_{i\in [n]}-|y| + (C^{T}x)_{i} \leq -|y| < - 1 + \frac{1}{4n},
\end{align*}
where the second step follows from $|y| = 1 -|x| > 1-\frac{1}{4n}$ by our assumption
\begin{align*}
    \max_{i\in [n]}-|x| + (Ry)_{i} &> -\frac{1}{4n} + \max_{i\in [n]} (Ry)_{i} \geq -\frac{1}{4n} + \frac{1}{n}\sum_{i}(Ry)_{i}  = -\frac{1}{4n} + \frac{1}{n}\sum_{j\in [n]}y_{j}\sum_{i\in [n]}R_{ij}\\
    &\geq -\frac{1}{4n} + \frac{1}{n}\sum_{j\in [n]}y_{j}\cdot(1 - n)\geq -\frac{1}{4n} + \frac{1}{n}\cdot (1 - n) = -1 + \frac{3}{4n}. 
\end{align*}
The first step follows from $|x| < \frac{1}{4n}$. We repalce max with average in the second step.
The fourth step comes from the fact that there exists at least one $0$ entry for each column of the payoff matrix $R$. 
Thus we have 
\begin{align*}
    \max_{i\in [n]}-|x| +(Ry)_{i} - \max_{i\in [n]}-|y| + (C^{T}x)_{i} > \frac{1}{2n} > \eps,
\end{align*}
which contradicts with the fact that $|y| > 1 - \frac{1}{4n}$. Therefore, we have $x > \frac{1}{4n}, y > \frac{1}{4n}$. Consequently, for any $i, j \in [n]$, $x_{i} > 0$, we have $(-|x| + (Ry)_i) - (-|x| + (Ry)_j) = (Ry)_{i} - (Ry)_j > \eps$. 
Hence, we have $(R\tilde{y})_i - (R\tilde{y})_j > 4n\eps$ for any $i, j \in [n]$ and $\tilde{x}_i > 0$. 
The same holds for the column player,
confirming that $(\tilde{x}, \tilde{y}) = (x/|x|, y/|y|)$ is a 4n$\eps$-Nash of the win-lose game $(R, C)$.
\end{proof}

Define $E = -A^{T} - I$ and $D = -A^{T} = E + I$; note that $E \in \{0, 1\}^{n}$ and the diagonal entries $E_{ii}$ are zero. 
Now we claim:
\begin{lemma}
\label{lem:pure-ppad}
Let $E$ be the adjacency matrix of the directed network of a public goods game (with divisible goods game and summation utility of the players). 
Then from any $\eps$-pure Nash equilibrium $\bs = (s_1, \ldots, s_n)$ of the public goods game, 
we can find a symmetric $3n\eps$-Nash of game $(A, B)$.  
\end{lemma}
\begin{proof}
Let $\bs = (s_1, \ldots, s_n)$ be an $\eps$-approximate pure Nash equilibrium of the public goods game, then for any agent $i \in [n]$, we have
$\sum_{j \in N(i)}s_j \geq 1 - \eps$. 
Otherwise, agent $i$ would increase its effort.
Moreover, we claim that $\sum_{j \in N(i)}s_j > 1 + \eps$ implies $s_{i} = 0 \pm \eps$. 
This holds because (1) if $\sum_{j \in N_i}s_i > 1$, then the best response of agent $i$ is $b_{i}(s_{-i}) =  0$, it then follows $x_{i} = 0 \pm \eps$; (2) $\sum_{j \in N_)}s_i \leq 1$, then the best response is $b_i(s_{-i}) = 1 - \sum_{j \in N)}s_i$ and $s_{i} - b_i(s_{-i}) = \sum_{i \in N(i)}s_{i} - 1 > \eps$, which contradicts with the equilibrium condition. In summary, for all $i \in [n]$, we have 
\begin{align*}
    (D^{T} \bs)_i &\geq 1 - \eps\\
    s_i = 0 \pm \eps &\text{ or } (D^{T}s)_i = 1 \pm \eps.
\end{align*}
Denote $\bs' = \max\{\bs - \eps, 0\}$, then we have
\begin{align*}
    (D^{T}\bs')_i \geq 1 - (n + 1)\eps\\
    \bs'_i = 0 \text{ or } (D^{T}\bs')_i = 1 \pm n\eps.
\end{align*}
Since $A = - D^{T}$, we have 
\begin{align*}
    (A\bs')_i \leq -1 + (n + 1)\eps\\
    \bs'_i = 0 \text{ or } (A\bs')_i = -1 \pm n\eps.
\end{align*}
Since $|\bs'| \geq \sum_{j \in N(1)}s'_{j} \geq 1 - (n + 1)\eps$, we we conclude that $\bs'/|\bs'|$ is a symmetric $3n\eps$-Nash of the game $(A, B)$
\end{proof}
Combining Lemma~\ref{lem:pure-ppad} and Lemma~\ref{lem:symmetric}, we conclude that it is PPAD-hard to find an $\eps$-approximate pure Nash equilibrium of public goods game, for $\eps = 1/\poly(n)$. 
This concludes the proof of Theorem~\ref{thm:ppad-pure}.

\subsection{Best-shot rule}
\label{sec:best-shot}
When the utility function is the best-shot rule (i.e., the utility of a node is the maximum of the provisions by its predecessors), there is a simple proof that there is no pure Nash equilibrium: First we prove that, in any pure Nash equilibrium, an agent plays either $0$ or $1$ (not any number between $(0, 1)$).
Then the result follows from the example shown in Figure~\ref{fig:pure-nash}. 

For mixed Nash equilibria, we have the following theorem, shown through a simple reduction from the indivisible case
\begin{theorem}
\label{thm:best-shot-ppad}
When the utility function is the best-shot rule (max), it is PPAD-hard to find a mixed Nash equilibrium of the public goods game.
\end{theorem}
\begin{proof}
	We set the valuation function to be $U(s) = \max\{1, s\}$.
We assume in an equilibrium profile, the player prefers a mixed combination over action $0, 1$ to a stategy $s \in (0,1)$ if they have the same utility guarantee.
We then prove that, in a mixed Nash equilibrium, an agent will only play a mixed strategy over actions 0 and 1. 
To see this, first, in an equilibrium profile, no player chooses to play $s$ with $s > 1$ in the support of its mixed strategy, since it can decrease it to $1$, which reduces the cost and does not affect the utility.
Next, if a player chooses to play $s \in (0,1)$ in the support of its mixed strategy, then we claim it is always better to replace $s$ with a convex combination of $0$ and $1$, i.e. chooses $0$ with probability $1-s$ and chooses $1$ with probality $s$.
We divide into two cases. (1) If the max production of neighbors is $s' \geq s$. Then the utility for later profile gets larger while the cost remains the same. (2) If the max production of neighbors is $s' < s$, then both the cost and utility remains the same.

Assuming all agent play mixed strategies over actions 0 and 1 in the equilibrium profile, it is not hard to modify the proof of Theorem~\ref{thm:ppad-mix} to show that it is PPAD-hard to find a mixed Nash equilibrium. We conclude the proof here.
\end{proof}

	\section{The Bounded Treewidth Algorithm}
\label{sec:bounded-treewidth}

When the treewidth of the underlying graph is bounded by $O\left(\frac{\log n}{\log\log n}\right)$, we develop a PTAS for computing an $\eps$-Nash of the (indivisible) public goods game.
We first recall the definition of {\em tree decomposition}.

\begin{definition}
\label{def:treewidth}
(Tree decomposition)
A tree decomposition\footnote{In defining treewidth, we ignore directions of the edges.} of a graph $G(V, E)$ is a tree $T$, with nodes $X_1, \ldots X_{|T|}$. Each node $X_i$ is a subset of $V$, and it satisfies:
\begin{enumerate}
    \item The union of $X_i$ equals $V$.
    \item For each edge $(u, v)\in E$, there exists a node $X_i$ that contains both vertices $u$ and $v$.
    \item For any vertex $u \in V$, the set of tree nodes that contain the vertex $u$ forms a connected sub-tree of $T$.
\end{enumerate}
The {\em width} of a tree decomposition is defined as $\max_{1\leq i\leq |T|}|X_i| - 1$ and the {\em treewidth} of a graph $G$, denoted as $\twd(G)$, is the minimum width among all tree decompositions of the graph $G$.
\end{definition}

We will call the vertices of $T$  {\em nodes}, and those of $G$ {\em vertices}. 
The treewidth $\twd(G)$ will be abbreviated by $w$, while $d$ is the maximum degree of $G$.  Our main result is shown below. Comparing with the general result of~\cite{daskalakis2006computing}, we get rid of the exponential dependence on $d$. Alas, we make no assumptions on the sparsity of the graph.
\begin{theorem}
\label{thm:ptas}
Given an indivisible public goods game defined on a network $G(V, E)$, we can find an $\eps$-Nash equilibrium in time $\poly(n) \cdot\min\{2d/\eps, 16\log (n) / \eps \}^{O(w)}$.
In particular, when the treewidth is $O\left(\frac{\log n}{\log\log n}\right)$, we can find an $\eps$-Nash equilibrium in $\poly(n) \cdot\left(\frac{1}{\eps}\right)^{O(w)}$ time.
\end{theorem}

First, a few notes about the proof.  
The time complexity of our algorithm depends minimally on $d$, while $d$ is in the exponent of the algorithm in~\cite{daskalakis2006computing}.  
To achieve this, we need to circumvent several difficulties, explained below.   Like the proof in \cite{daskalakis2006computing}, we first need to show the existence of an approximate Nash equilibrium with probabilities that are multiples of a small real $\delta>0$.  
Simply applying the total variation bound gives $\delta = O(\frac{\eps}{d})$, which is not coarse enough.  
In Lemma~\ref{lem:discretize-existence}, we use a probabilistic argument showing the existence of an approximate Nash equilibrium after discretizing the strategy space.
In particular, we randomly round a Nash equilibrium, for $\delta = O(\frac{\eps}{\log n})$ and utilize the concentration property.
Another difficulty is that the algorithm~\cite{daskalakis2006computing} works on the {\em primal} graph (see~\cite{daskalakis2006computing} for the definition), whose treewidth can be $w\cdot d$, yielding an exponential dependence on $d$. 
Instead, our algorithm directly works on the the original graph through {\em dynamic programming}, with no exponential dependence on $d$. 
Our dynamic programming method bares some similarities with the approach in~\cite{thomas2015pure}.  However, we must modify significantly that algorithm, whose running time has a polynomial dependency on the size of the payoff matrix, which in our case be exponential. 

Now, to prove the theorem, by Lemma~\ref{lem:threshold-to-public}, it suffices to show how to compute an $\eps$-approximate equilibrium of a threshold game $\g(V, E, t)$. 
Again, we assume $t = 1/2$ for simplicity. 
We discretize the strategy space of each player to $S^{\delta} = [\delta]$, where $\delta = \max\{\eps /2d, \eps / 16\log n\}$. 
We first show that there exists an $\eps$-approximate {\em pure} Nash equilibrium in strategy space.
\begin{lemma}
\label{lem:discretize-existence}
For any threshold game $\g(V, E, \frac{1}{2})$, there exists an $\eps$-approximate equilibrium when we restrict the strategy space to be $[\delta]^{n}$, where $\delta =  \eps / 16\log n$.
\end{lemma}
\begin{proof}
Suppose $\bx = (x_1, \ldots, x_n)$ is an equilibrium profile of the threshold game $\g(V, E, \frac{1}{2})$. For any $i \in [n]$, suppose $x_i \in [t_i\delta, (t_{i} + 1)\delta]$, then we randomly round $x_i$ to $\tx_i \in \{t_i\delta, (t_{i} + 1)\delta\}$, and we have
\begin{align*}
    \tx_i = \left\{\begin{matrix}
    \left(t_i + 1\right)\delta & \text{with prob. }\frac{x_i}{\delta} - t_i\\
    t_{i}\delta & \text{with prob. }1 - \frac{x_i}{\delta} + t_i.
    \end{matrix}\right.
\end{align*}
We remark that $\E[\tx_i] = x_i$ and $(\tx_i - t_i \delta)$ is a binary random variable that takes value in $\{0, \delta\}$, with mean $(x_i - t_i\delta)$. The rest of the proof establishes that $(\tx_1, \ldots, \tx_n)$ is an $\eps$-approximate equilibrium with positive probability, therefore proving its existence. 

By the multiplicative Chernoff bound, for any $i \in [n]$, if $\sum_{j\in N_i}(x_i - t_i\delta) = \sum_{j\in N_i}(\E[\tx_i] - t_i\delta) < \eps$, then we have
\begin{align}
    \label{eq:ptas1}
    \Pr\left(\sum_{j\in N_i}\tx_{j} - \sum_{j\in N_i}x_j \leq -\eps \right) = 0
\end{align}
and
\begin{align}
    \label{eq:ptas2}
    \Pr\left(\sum_{j\in N_i}\tx_{j} - \sum_{j\in N_i}x_j \geq \eps \right) &=
    \Pr\left(\sum_{j\in N_i}\left(\tx_{j} - t_j\delta\right) -\sum_{j\in N_i}\left(\E[\tx_j] - t_j\delta\right) \geq \eps\right)
    \leq \exp\left(-\frac{\eps}{3\delta}\right) \leq n^{-2}.
\end{align}
If $\sum_{j\in N_i}(x_i - t_i\delta) = \sum_{j\in N_i}(\E[\tx_i] - t_i\delta) \in [\eps, 2]$, then we have
\begin{align}
    \Pr\left(\left|\sum_{j\in N_i}\tx_{j} - \sum_{j\in N_i}x_j\right| \geq \eps \right) &= 
    \Pr\left(\left|\sum_{j\in N_i}\left(\tx_{j} - t_j\delta\right) -\sum_{j\in N_i}\left(\E[\tx_j] - t_j\delta\right)\right| \geq \eps\right) \notag\\
    &\leq 2\exp\left(-\frac{\eps^2}{2\delta \left(\sum_{j\in N_i}\E[\tx_j] - \sum_{j\in N_i}t_i\delta\right)}\right) \leq 2\exp\left(-\frac{\eps}{4\delta}\right) \leq n^{-2}.\label{eq:ptas3}
\end{align}
If $\sum_{j\in N_i}(x_i - t_i\delta) = \sum_{j\in N_i}(\E[\tx_i] - t_i\delta) > 2$, it then follows that $x_i = 0$ and we have
\begin{align}
\label{eq:ptas4}
     \Pr\left(\sum_{j\in N_i}\tx_{j} < 1 \right) \leq \Pr\left(\sum_{j\in N_i}\left(\tx_{j} - t_j\delta\right) < 1\right) \leq \exp\left(-\frac{1}{4\delta}\right) \leq n^{-2}.
\end{align}
Combining Eq.~\eqref{eq:ptas1}~\eqref{eq:ptas2}~\eqref{eq:ptas3}~\eqref{eq:ptas4} and using an union bound, 
we conclude that $(\tx_1,\ldots, \tx_n)$ satisfies equilibrium condition with probability at least $(1 - 1/n)$, completing the proof.
\end{proof}

We next provide an algorithm that finds an $\eps$-approximate equilibrium based on dynamic programming. 
A {\em nice tree decomposition} is a tree decomposition $T$ that only contains the following four types of nodes (see Figure~\ref{fig:nice-tree} for an illustration).
\begin{enumerate}
    \item Leaf node.
    \item Forget node. Such a node $i$ has only one child $i'$, and $X_{i'} = X_{i} \backslash \{v\}$ for some vertex $v \in V$.
    \item Introduce node. Such a node $i$ has only one child node $i'$, and $X_{i'} = X_{i} \cup \{v\}$ for some vertex $v \in V$.
    \item Join node. Such a node $i$ has exact two children nodes $i_1, i_2$, and $X_i = X_{i_1} = X_{i_2}$.
\end{enumerate}

Any tree decomposition can be converted into a nice tree decomposition, of size at most $w\cdot |V|$, in linear time without enlarging the width~\cite{bodlaender2008combinatorial}.
\begin{figure}[!h]
	\centering
	\begin{minipage}[t]{.3\textwidth}
		\centering
		\begin{tikzpicture}
		\node (0) at (0, 1.5) [circle, draw, very thick] {$u, v, w$};
		\node (1) at (0, 0) [circle, draw, very thick] {$u, v$};

		\draw[very thick] (0) -- (1);

		\end{tikzpicture}
		\caption*{Forget}
		\label{pic:forget}
	\end{minipage}%
	\begin{minipage}[t]{.3\textwidth}
		\centering
		\begin{tikzpicture}
    		\node (0) at (0, 1.5) [circle, draw, very thick] {$u, v$};
    		\node (1) at (0, 0) [circle, draw, very thick] {$u, v, w$};
    
    		\draw[very thick] (0) -- (1);
		\end{tikzpicture}
		\caption*{Introduce}
		\label{pic:introduce}
	\end{minipage}
	\begin{minipage}[t]{.3\textwidth}
		\centering
		\begin{tikzpicture}
		    \node (0) at (0, 1.5) [circle, draw, very thick] {$u, v$};
    		\node (1) at (-1, 0) [circle, draw, very thick] {$u, v$};
    		\node (2) at (1, 0) [circle, draw, very thick] {$u, v$};
    
    		\draw[very thick] (0) -- (1);
    		\draw[very thick] (0) -- (2);

		\end{tikzpicture}
		\caption*{Join}
		\label{pic:join}
	\end{minipage}
	\caption{An illustration for three types of nodes of a nice tree decomposition.}
	\label{fig:nice-tree}
\end{figure}
	
Now, we have
\begin{lemma}
\label{lem:dp-nash}
Given a threshold game $\g(V, E, \frac{1}{2})$ with the strategy space $[\delta]^{n}$, and a nice tree decomposition $T$ of the graph $G(V, E)$, we can compute an $\eps$-approximate equilibrium in $\delta^{-O(w)}$ time.
\end{lemma}
\begin{proof}
Given a nice tree decomposition $T$, we compute an $\eps$-approximate equilibrium via a bottom-up approach. 
For any node $X_{i} \in T$, we use $V_{i}$ to denote all vertices contained in $X_i$ and its sub-tree.  
We compute a table $T_i: [\delta]^{|X_i|}\times [\delta]^{|X_i|} \rightarrow \{0, 1\}$ for each node $X_i$, and we note that the size of the table is bounded by $\delta^{-O(w)}$.
Ideally, we would set an entry $T_{i}(s_1, \ldots, s_{|X_i|}, c_1, \ldots, c_{|X_i|}) = 1$, iff there exists a strategy profile $(p_1, \ldots, p_{|V_i|})$ of vertex set $V_i$, such that 
\begin{enumerate}[(i)]
\item for any vertex $v \in V_i\backslash X_i$, the vertex $v$ satisfies the equilibrium condition, 
\item for any vertex $v\in X_i$, $p_v = s_v$ and $\sum_{j \in N_v\cap (V_i \backslash X_i)}p_{v} = c_{v}$, i.e., the summation of vertex $v$'s neighbor in $V_i \backslash X_i$ is $c_v$. 
\end{enumerate}
We remark that vertices in $X_i$ do not need to satisfy the equilibrium condition, and we only record the summation of their neighbors in $V_i \backslash X_i$. 
Next, we show how to do update the table in a bottom-up manner.

(1) Leaf. 
For any leaf $X_i \in T$ and $s, c \in [\delta]^{|X_i|}$, we set $T_i(s, c) = 1$ if and only if $c = (0,\ldots, 0)$.

(2) Forget. 
Suppose $X_{i'} =  X_{i}\cup \{v\}$ is the parent node, $s, c \in [\delta]^{|X_i|}$ and $s_v, c_v \in [\delta]$, we set $T_{i'}(s, s_v, c, c_v) = 1$ if $T_{i}(s, c) = 1$ {\em and} $c_v = 0$; we set $T_{i'}(s, s_v, c, c_v) = 0$ otherwise.

(3) Introduce. 
Suppose $X_{i'} = X_{i} \backslash \{v\}$ is the parent node and $s, c \in [\delta]^{|X_{i'}|}$, we set $T_{i'}(s, c) = 1$ iff there exists $(s, s_v, c, c_v) \in [\delta]^{2|X_i|}$, such that $T_{i}(s, s_v, c, c_v) = 1$ {\em and} the vertex $v$ satisfies the equilibrium condition, i.e., (i) if $c_v + \sum_{j \in N_{v}\cap X_i} s_{j} > \frac{1}{2}$, then $s_v = 0 \pm \eps$; (ii) $c_v + \sum_{j \in N_{v}\cap X_i}s_j < \frac{1}{2}$, then $s_v = 1 \pm \eps$.   

(4) Join. 
Suppose node $X_i$ has two children, $X_{i_1}$, $X_{i_2}$, and $X_{i} = X_{i_1} = X_{i_2}$. 
Then for any $s, c \in [\delta]^{|X_i|}$, we set $T_{i}(s, c) = 1$ iff there exists $c_1, c_2 \in [\delta]^{|X_i|}$, such that $T_{i_1}(s, c_1) = 1$, $T_{i_2}(s, c_2) = 1$, {\em and} for any $j\in X_i$, $c[j] = \min\{c_{1}[j] + c_{2}[j], 1\}$. 

After we reach the root $r$ and complete the table $T_r$, we verify equilibrium conditions for all vertices $v \in X_{r}$. 
To be more specific, if there exists a configuration $(s, c)\in [\delta]^{|X_i|} \times [\delta]^{|X_i|}$, 
such that $T_r(s, c) = 1$ {\em and} all vertices $v$ in $V_r$ satisfy the equilibrium, i.e., 
(i) if  $c_v + \sum_{j\in N_{v}\cap V_r}s_j> \frac{1}{2} + \eps $ then $s_v = 0\pm \eps$; 
(ii) if $c_v + \sum_{j\in N_{v}\cap V_r}s_j < \frac{1}{2} - \eps$ then $s_v = 1\pm \eps$;
we then confirm that there exists an $\eps$-approximate equilibrium. 
We can find one by either fixing the strategy of all vertices $v\in V_r$ to be $s_v$, 
and recursively computing equilibrium profiles in the sub-tree; 
or we can associate a satisfiable assignment (if there exists one) for each entry during the dynamic programming process.
We output that there is no $\eps$-approximate equilibrium profile otherwise.
\end{proof}
Combining Lemma~\ref{lem:dp-nash} and Lemma~\ref{lem:discretize-existence}, we conclude the proof of Theorem~\ref{thm:ptas}.

We can show a similar result for divisible public good games with the summation rule. Again, when the treewidth of the underlying graph is bounded by $O(\log n/ \log\log n)$, there is a PTAS for finding an $\eps$-approximate pure Nash equilibrium of the public goods game: 
\begin{theorem}
	\label{thm:ptas-sum2}
	Given a divisible public goods game with summation utility defined on a directed network $G(V, E)$, we can find an $\eps$-approximate pure Nash equilibrium in time $\poly(n)\cdot\min\{2d/\eps, 16\log (n / \eps) \}^{O(w)}$ time.
	In particular, when the treewidth is $O(\log n /\log\log n)$, we can find an $\eps$-approximate pure Nash equilibrium in $\poly(n) \cdot (\frac{1}{\eps})^{O(w)}$ time.
\end{theorem}
The proof is similar to Theorem~\ref{thm:ptas}, and is omitted.

	\section{Discussion}
\label{sec:discussion}
We explore the complexity of equilibria in public goods games played on directed graphs. One striking conclusion is the ubiquity of PPAD-completeness in this domain.  For a number of quite different reasons, very different variants of the problem are shown to share the same fate --- and a rather sophisticated fate at that.  
This is in stark contrast with the corresponding public goods games  in undirected graphs, where the consensus is that equilibria are rather boring (but see the discussion below of some intriguing problems in undirected networks raised by this work). 
Note that graphical games are already intractable when they are symmetric --- but, of course, this is because the local normal form games in each neighborhood can simulate any asymmetry. 

Are public good games in directed networks real?  It can be argued\footnote{Many thanks to the reviewers for bringing this point up.} that some of the directed graphs we evoked in the introduction (towns that are downwind or upriver from one another, or the relationship ``B is on A's way to work'') are {\em transitive}, and it is easy to see that public good games on such directed graphs have trivial equilibrium problems. On the other hand, many social networks with sharing features are asymmetric and non-transitive, and so are infection networks in much of epidemic modeling. In addition, we hope that our techniques and techniques may be a public good of some value to the community.

For the indivisible case, we found that there are three special cases of utilities that admit polynomial time solution: Flat utilities, steep utilities, plus a third polynomial case, alternating utilities, which is quite unexpected and intriguing (its algorithm relies on the solution of a system of equations in $\bf F_2$).  We show that these are the only tractable cases.  But there is an interesting {\em variant} of this problem which is quite mysterious:  Suppose that we allow the utility function to be such that {\em certain steps of $X$ have height exactly $p$}, and therefore nodes can be indifferent between buying the good and free-riding.  We suspect that this variant is subject to the same dichotomy, but it seems much harder to prove.   Consider for example the function $X(1) = 1>p, X(k)= 1+p$ for all $k>1$.  Then it is easy to see that, in this case, odd cycles {\em do} have an equilibrium, with all players producing the good: the $p$ step makes them indifferent to doing so.  This deprives us of a valuable gadget.  It turns out that there is a 7-node, 21-edge gadget with no equilibrium for this case: the node set is $\{1,\ldots, 7\}$ and the edges go from $i$ to $i+1, i+2, 1+4\bmod 7$. But this does not immediately give us an NP-hardness proof, nor does it generalize to other composition functions with $p$ steps.  

We believe that the general form of the utility function of public goods games in networks (Section 3) has not been articulated in the past, and it does lead to an interesting complexity classification problem.  We note that the equilibrium situation in this generality is open in the {\em undirected} case.  Take for example the case where $p<1$ and $X$ is the threshold function $X(i) = 1$ if $i\geq 2$ and zero otherwise; we believe that this case is NP-complete.  The complexity dichotomy problem in undirected networks with general utility functions is a very interesting open problem raised by this work.


The divisible good games under the summation utility are something of a mystery when it comes to {\em mixed} equilibria.  As with other games with uncountable strategy spaces, it is not easy to characterize mixed Nash equilibria in a tangible, useful way.  We believe that positive results may be possible here:  Could it be that there are always mixed Nash equilibria with small support, and in fact they are easy to find?   There are reasons for hope for a truly positive result in this case.

The intractability of simple Nash equilibrium problems in common goods games in directed networks is an indication that asymmetry in social systems --- a notion intuitively coterminous with unfairness --- may consistently lead to instability.  Can the intractability proofs help identify the features of the directed networks, and of the agents and their utilities, which are at the root of such instability?  This could lead to principles for better design of social networks, or beneficial interventions therein.

	\begin{acks}
		The authors would like to thank Xi Chen and two anonymous EC reviewers for their very helpful feedback.This research was supported by NSF grants CCF-1763970 AF and CCF-1910700 AF, and a grant from Softbank.
	\end{acks}

	\bibliographystyle{ACM-Reference-Format}
	\bibliography{ref}
	
	\ifdefined\isArxiv
	\newpage
	\appendix
	\section{Omitted proof from Section~\ref{sec:ppad-hard-mix}}
\label{sec:omit-proof}
\subsection{Missing proof from Section~\ref{sec:equivalence}}
\label{sec:equivalence-app}
\lemReduction*
\begin{proof}
    We first reduce the threshold game to the public good game.
    Given an instance of the threshold game $\g(V, E, t)$, we construct a public good game as follow. 
    We keep the network $G(V, E)$ unchanged and set the value of the good to be $U = 1$ and the price to be $p = e^{-t} \in (0, 1)$. 
    For any $\eps$-Nash $\bs = (s_1, \cdots, s_n)$ of the public good game, 
    we construct an $\eps$-approximate equilibrium $\bx = (x_1, \cdots, x_n)$ of $\g(V, E, t)$ as
    \begin{align*}
        x_i = \min\{-\log(1 - s_j), 1\} \in [0, 1], \forall i.
    \end{align*}
    Consider any agent $i$ in the public good game, its utility is specified as
    \begin{align*}
        U(s_i, s_{-i}) = \left\{\begin{matrix}
        1 - p & s_i = 1\\
        1 - \prod_{j \in N_i}(1 - s_j) & s_i = 0,
        \end{matrix}\right.
    \end{align*}
    thus we have 
    \begin{align*}
    U(1, s_{-i}) - U(0, s_{-i}) = \prod_{j \in N_i}(1 - s_j) - p.
    \end{align*}
    We divide into three cases.
    
    Case 1. $\prod_{j \in N_i}(1 - s_j) - p > \eps$. This implies $s_i = 1$ and $x_{i} =  \min\{-\log(1 - s_i), 1\} = 1$. Now we have
    \begin{align*}
        &\prod_{j \in N_i}(1 - s_j) - p > \eps \Rightarrow \prod_{j \in N_i}(1 - s_j) > p + \epsilon 
        \Rightarrow \log \prod_{j \in N_i}(1 - s_j) > \log (p + \epsilon)\\
        &\Rightarrow  \sum_{j\in N_i} -\log (1-s_j) < -\log(p+\epsilon) < -\log p = t. 
    \end{align*}
    Since $\max_{j\in N_i}\{-\log(1-s_j)\} \leq  \sum_{j\in N_i} -\log (1-s_j) < t < 1$, we have $\sum_{j\in N_i} x_j = \sum_{j\in N_i} -\log (1-s_j) < t$, this satisfies the equilibrium condition of the threshold game.
    
    Case 2. $\prod_{j \in N_i}(1 - s_j) - p < -\eps$. This implies $s_{i} = 0$ and $x_i = 0$. Similar to the first case, we have
    \begin{align*}
        &\prod_{j \in N_i}(1 - s_j) - p < -\eps \Rightarrow \prod_{j \in N_i}(1 - s_j) < p - \epsilon 
        \Rightarrow \log \prod_{j \in N_i}(1 - s_j) < \log (p - \epsilon)\\
        &\Rightarrow  \sum_{j\in N_i} -\log (1-s_j) > -\log(p-\epsilon) > -\log p = t. 
    \end{align*}
    Since $t < 1$ and $-\log (1-s_j) > 0$ for $\forall j \in N_i $, we conclude that $\sum_{j\in N_i} x_j = \sum_{j\in N_i} \min\{-\log (1-s_j), 1\} > t$, this satisfies the equilibrium condition of the threshold game.
    
    Case 3 $\prod_{j \in N_i}(1 - s_j) - p \in [-\eps, \eps]$. This time $s_{i}$ can be any number in $[0, 1]$, so does $x_{i}$. We need to verify that $\sum_{j \in N_i}x_i \in [t - 8\eps, t + 8\eps]$. We have
    \begin{align*}
    &\prod_{j \in N_i}(1 - s_j) - p \in [-\eps, \eps] \Rightarrow \prod_{j \in N_i}(1 - s_j) \in [p-\eps, p+\eps]\\ 
    &\Rightarrow \sum_{j\in N_i} -\log (1-s_j) \in [-\log(p + \eps), -\log(p - \eps)].
    \end{align*}
    When $\eps < \eps < \min\{0.1,\frac{t}{8},\frac{1-t}{8}\}$, we can prove that $[-\log(p + \eps), -\log(p - \eps)] \in [t - 8\eps, t + 8\eps]$. 
    We defer the calculation to lemma~\ref{lem:calculate1}. 
    Now we have $\sum_{j\in N_j}x_{j} = \sum_{j\in N_i} \min\{-\log (1-s_j), 1\} = \sum_{j\in N_i} -\log (1-s_j) \in [t - 8\eps, t + 8\eps]$, which satisfies the equilibrium condition.
    
    We next show there is a polynomial time reduction from public good games to threshold games. 
    Similar as above, given an instance of public good game defined on $G(V, E)$, $U = 1, 0 < p < 1$, we construct a threshold game on the same network $(V, E)$, with $t = \frac{1}{2}$. 
    Given an $\eps$-approximate equilibrium $\bx = (x_1, \ldots, x_n)$ of the threshold game, we recover an $-4p\log(p)\eps$-Nash $\bs = (s_1, \ldots, s_n)$ of the public good game as follow,
    \begin{align*}
        s_{i} = \left\{\begin{matrix}
        1 - p^{2x_i} & x_i \leq \frac{1}{2} + \eps\\
        1 & \text{otherwise}.
        \end{matrix}\right.
    \end{align*}
    For any agent $i$, if $\sum_{j \in N_i}x_j > \frac{1}{2}+ \eps$, then $x_i = 0$ and $s_i = 0$ by definition. 
    It then follows that $U(1, s_{-i}) - U(0, s_{-i}) = \prod_{j \in N_i}(1 - s_j) - p \leq p^{\sum_{j \in N_i}2x_j} - p \leq p^{1 + 2\eps} - p < 0$. 
    Hence, it satisfies the equilibrium condition. If $\sum_{j \in N_i}x_j < \frac{1}{2} - \eps$, then $x_i = 1$ and $s_i = 1$. Meanwhile, we have $U(1, s_{-i}) - U(0, s_{-i}) = p^{\sum_{j \in N_i}2x_j} - p > p^{1 - 2\eps} - p > 0$. 
    Finally, if $\sum_{j \in N_i}x_j \in [\frac{1}{2}-\eps, \frac{1}{2}+\eps]$, we have $U(1, s_{-i}) - U(0, s_{-i}) = \prod_{j \in N_i}(1 - s_j) - p = p^{\sum_{j \in N_i}2x_j} - p \in [p(p^{2\eps} - 1), p(p^{-2\eps} - 1)] \in [2p\log (p)\eps, -4p\log (p)\eps]$. 
    Here we use the facts that $\lambda \leq e^{\lambda} - 1\leq 2\lambda$ for $\lambda < 1$. 
    Therefore, we have verified that $\bs=(s_1, \ldots,s_n)$ is an $-4p\log(p)\eps$-Nash of the public good game. 
    Hence, setting $c_p = -4p\log p$, we conclude the proof.
\end{proof}
\begin{lemma}
\label{lem:calculate1}
For any $0 < t < 
1$ and $0 < \eps < \min\{0.1,\frac{t}{8},\frac{1-t}{8}\}$, we have
\begin{enumerate}
    \item $-\log(e^{-t} - \eps) < t + 8\eps$,
    \item $-\log(e^{-t} + \eps) > t - 8\eps$.
\end{enumerate} 
\end{lemma}
\begin{proof}
    We have
    \begin{align*}
        &-\log(e^{-t} - \eps) < t + 8\eps \Leftrightarrow \log(e^{-t} - \eps) > -(t + 8\eps) \Leftrightarrow e^{-t} - \eps > e^{-(t + 8\eps)}\\
        &\Leftrightarrow e^{-t}\left(1 - e^{-8\eps}\right) > \eps \Leftarrow 1 - e^{-8\eps} > 3\eps.
    \end{align*}
    By simple calculations, we can show $1 - e^{-8\eps} > 3\eps$ for $\eps < 0.1$. On the other side, we have
    \begin{align*}
    &-\log(e^{-t} + \eps) > t - 8\eps \Leftrightarrow \log(e^{-t} + \eps) < -(t - 8\eps) \Leftrightarrow e^{-t} + \eps < e^{-t}\cdot e^{8\eps} \Leftrightarrow e^{-t}(e^{8\eps} - 1)\geq \eps
    \end{align*}
    This follows from the fact that $e^{-t}(e^{8\eps} - 1) \geq \frac{1}{3}(e^{8\eps} - 1) \geq \frac{8}{3}\eps > \eps$.\qedhere
\end{proof}

\subsection{Missing proof from Section~\ref{sec:reduce-circuit}}
\label{sec:reduce-circuit-app}

We provide PPAD-hardness proof for threshold game.

\begin{theorem}[Restatement of Theorem~\ref{thm:threshold-hard}]
\label{thm:threshold-hard-app}
It is PPAD-hard to find an $\eps$-approximate equilibrium of the threshold game, for some constant $\eps > 0$.
\end{theorem}

\begin{proof}
We fix the threshold $t = \frac{1}{2}$ in the rest of the proof. Furthermore, we restrict the equilibrium strategy in $[0, \frac{1}{2} + \eps] \cup \{1\}$, since for any $\eps$-approximate equilibrium $\bx = (x_1, \ldots, x_n)$, we could set 
\begin{align*}
\tilde{x}_{i} = \left\{\begin{matrix}
x_{i} &  x_{i} \leq \frac{1}{2} + \eps\\
1 & \text{otherwise}
\end{matrix}
\right.
\end{align*}
and we can easily verify that $\tilde{\bx} = [\tilde{x}_1, \ldots, \tilde{x}_n]$ is still an $\eps$-approximate equilibrium. 
We use the strategies of players in the threshold game to represent an $\eps$-approximate assignment to $\eps$-$\gc$, and build all 9 types of gates in $ \{G_{\xi}, G_{\times \xi}, G_{=}, G_{+}, G_{-}, G_{<}, G_{\wedge}, G_{\vee}, G_{\neg}\}$. 
We start from constructing an elementary game gadget $G_{\frac{1}{2}-}(\nm{v_1, v_2} v)$ (see Figure~\ref{fig:element-gadget}), where $v_1, v_2 \in V\cup \{nil\}$ are input players, $v\in V$ is the output player. 
The output player $v$ could have many out-coming edges, but it only has one in coming edge from the internal node $v_b$. 
The elementary gadget $G_{\frac{1}{2}-}(\nm{v_1, v_2} v)$ serves as a building block for later constructions and proves useful throughout our proof. 
Ideally, it poses the constraint that $\bx[v] = \max\{\frac{1}{2} - \bx[v_1] - \bx[v_2], 0\}$ in an equilibrium.
\begin{figure}[!h]
    \centering
	\begin{tikzpicture}
		\node (0) at (0, 0) [circle, draw, very thick] {$v_a$};
		\node (1) at (1.73, 1) [circle, draw, very thick] {$v_b$};
		\node (2) at (0, 2) [circle, draw, very thick, scale=1.2] {$v$};

		\draw[->,  very thick] (0) -- (1);
		\draw[->,  very thick] (1) -- (2);
		\draw[->,  very thick] (2) -- (0);
		
		\draw[->,  very thick] (2) -- (0,2.8);
		\draw[->,  very thick] (2) -- (-0.6,2.5);
		\draw[->,  very thick] (0,-0.8) -- (0);
		\draw[->,  very thick] (-0.6,-0.5) -- (0);
		
		\node[text width=1cm, scale=1] at (0.4, -1) {$v_1$};
		\node[text width=1cm, scale=1] at (-0.4, -0.6) {$v_2$};

	\end{tikzpicture}
    \caption{Elementary gadget}
    \label{fig:element-gadget}
\end{figure}
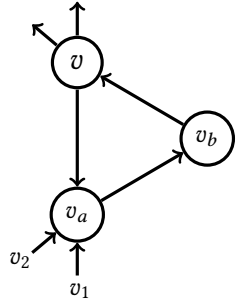
\begin{lemma}
\label{lem:element-gadget}
Consider the game gadget $G_{\frac{1}{2}-}( \nm{v_1, v_2} v)$ constructed in Figure~\ref{fig:element-gadget}. 
In any $\eps$-approximate equilibrium of the threshold game $\g(V, E, \frac{1}{2})$, we have $\bx[v] = \max\{\frac{1}{2} - \bx[v_1] - \bx[v_2], 0\} \pm \eps$. 
In particular, if we set $v_2 = nil$, then $\bx[v] = \max\{\frac{1}{2} - \bx[v_1], 0\} \pm \eps$; if we set $v_1, v_2 = nil$, then $\bx[v] = \frac{1}{2} \pm \eps$.
\end{lemma}
\begin{proof}
Consider the in-coming neighbors of player $v_a$, $N_{a} = \{v, v_1, v_2\}$. 
In an $\eps$-approximate equilibrium, if $\bx[v] + \bx[v_1] + \bx[v_2] > \frac{1}{2}+\eps$, we have $\bx[v_a] = 0\pm \eps$ and $\bx[v_b] = 1 \pm \eps$. 
It then follows that $\bx[v] = 0\pm \eps$. 
This implies $\bx[v_1] + \bx[v_2] > \frac{1}{2}$, and thus $\max\{\frac{1}{2} - \bx[v_1] - \bx[v_2], 0\} \pm \eps = 0\pm \eps$. 
This satisfies the equilibrium condition. 
If $\bx[v] + \bx[v_1] + \bx[v_2] < \frac{1}{2}-\eps$, then we have $\bx[v_a] = 1\pm \eps$ and $\bx[v_b] = 0 \pm \eps$, this again implies $\bx[v] = 1 \pm \eps$. 
This contradicts with the fact that $\bx[v] + \bx[v_1] + \bx[v_2] < \frac{1}{2}-\eps$. 
In summary, we have $\bx[v] = \max\{\frac{1}{2} - \bx[v_1] - \bx[v_2], 0\} \pm \eps$.\qedhere
\end{proof}

Now we are ready to construct $G_{=}, G_{+}, G_{-}, G_{\frac{1}{2}}$ and $G_{\times \frac{1}{2}}$.  We assume the inputs of these gates belong to $[0, \frac{1}{2} + \eps]$, {\em except} for the COPY gate.
This assumption is not a loss of generality since: 
(1) if the input node $v_1$ is also the output node of another gate, then it is value is guaranteed to be in $[0, \frac{1}{2}+ \eps]$ by our construction below; (2) otherwise, we can always apply a COPY gate to restrict its value in $[0, \frac{1}{2}+ \eps]$.

(1) COPY $G_{=}(\nm{v_1} v)$. 
Concatenating $G_{\frac{1}{2}-}(\nm{v_1} v_2)$ with $G_{\frac{1}{2}-}(\nm{v_2} v)$, 
then we have $\bx[v_2] = \max\{\frac{1}{2} - \bx[v_1], 0\} \pm \eps$, and $\bx[v] = \max\{\frac{1}{2} - \bx[v_2], 0\} \pm \eps = \min\{\bx[v_1], \frac{1}{2}\} \pm 2\eps$.

(2) ADD $G_{+}(\nm{v_1, v_2} v)$. 
Concatenating $G_{\frac{1}{2}-}(\nm{v_1, v_2} v_3)$ with $G_{\frac{1}{2}-}(\nm{v_3} v)$, 
then we have $\bx[v_3] = \max\{\frac{1}{2} - \bx[v_1] - \bx[v_2], 0\} \pm \eps$ and $\bx[v] = \max\{\frac{1}{2} - \bx[v_3], 0\} \pm \eps = \min\{\bx[v_1] + \bx[v_2], \frac{1}{2}\} \pm 2\eps$.

(3) SUBTRACT $G_{-}(\nm{v_1, v_2} v)$. 
Concatenating $G_{\frac{1}{2}-}(\nm{v_1} v_3)$ with $G_{\frac{1}{2}-}(\nm{v_2, v_3} v)$, 
then we have $\bx[v_3] = \max\{\frac{1}{2} - \bx[v_1], 0\} \pm \eps$ and $\bx[v] = \max\{\frac{1}{2} - \bx[v_2] - \bx[v_3], 0\} \pm \eps = \max\{\bx[v_1] - \bx[v_2], 0\} \pm 2\eps$. 

(4) VALUE $G_{\frac{1}{2}}(v)$. 
We can simply use $G_{\frac{1}{2}-}(\nm{v_1} v)$ with $v_1 = nil$, i.e., there is no input to the gadget.

(5) HALF $G_{\frac{1}{2}}(\nm{v_1}, v)$. 
The HALF gate is shown in Figure~\ref{fig:half-gadget}. 
In order to halve the value of the input player, we need to carefully compose 4 elementary gadgets. 
Ideally, in an equilibrium profile, we expect $\bx[v_a] = \frac{1}{2} - \bx[v_1]$, $\bx[v_b] = \bx[v_1] - \bx[v]$ and $\bx[v_c] = \frac{1}{2} - \bx[v_1] + \bx[v]$.
The fourth gadget poses the constraint that $\frac{1}{2} - \bx[v_1] + \bx[v] + \bx[v] = \frac{1}{2}$, i.e., $\bx[v_1] = 2\bx[v]$. Formally, we have
\begin{figure}[!h]
    \centering
	\begin{tikzpicture}[scale=0.6, every node/.style={scale=0.6}]
		\node (10) at (0, 0) [circle, draw, very thick] {};
		\node (11) at (1.73, 1) [circle, draw, very thick] {};
		\node (12) at (0, 2) [circle, draw, very thick] {};
	    \node[text width=1cm, scale=1.5] at (0,2.2) {$v_b$};
	    \node[text width=1cm, scale=2] at (1.5,1) {$2$};

		\draw[->,  very thick] (10) -- (11);
		\draw[->,  very thick] (11) -- (12);
		\draw[->,  very thick] (12) -- (10);
		
		\node (20) at (4, 0) [circle, draw, very thick] {};
		\node (21) at (5.73, 1) [circle, draw, very thick] {};
		\node (22) at (4, 2) [circle, draw, very thick] {};
	    \node[text width=1cm, scale=1.5] at (6.8,1) {$v_c$};
	    \node[text width=1cm, scale=2] at (5.5,1) {$3$};

		\draw[->,  very thick] (20) -- (21);
		\draw[->,  very thick] (21) -- (22);
		\draw[->,  very thick] (22) -- (20);
		
		\node (30) at (2, -3) [circle, draw, very thick] {};
		\node (31) at (3.73, -2) [circle, draw, very thick] {};
		\node (32) at (2, -1) [circle, draw, very thick] {};
	    \node[text width=1cm, scale=1.5] at (2.2,-3) {$v$};
	    \node[text width=1cm, scale=2] at (3.5,-2) {$4$};

		\draw[->,  very thick] (30) -- (31);
		\draw[->,  very thick] (31) -- (32);
		\draw[->,  very thick] (32) -- (30);
		
        \node (40) at (-2, -3) [circle, draw, very thick] {};
		\node (41) at (-3, -1.27) [circle, draw, very thick] {};
		\node (42) at (-1, -1.27) [circle, draw, very thick] {};
	    \node[text width=1cm, scale=1.5] at (-0.52, -0.78) {$v_a$};
	    \node[text width=1cm, scale=2] at (-1.2, -1.9) {$1$};

		\draw[->,  very thick] (40) -- (41);
		\draw[->,  very thick] (41) -- (42);
		\draw[->,  very thick] (42) -- (40);
		
		\draw[->,  very thick] (30) -- (10);
		\draw[->,  very thick] (12) -- (22);
		\draw[->,  very thick] (21) -- (31);
		\draw[->,  very thick] (42) -- (10);
		
	    \draw[->,  very thick] (-2,-4) -- (40);
	    \node[text width=1cm, scale=1.5] at (-1.5, -4.3) {$v_1$};
	    \draw[->,  very thick] (30) -- (2,-4);

	\end{tikzpicture}
    \caption{HALF gate}
    \label{fig:half-gadget}
\end{figure}
\begin{lemma}
\label{lem:half-gadget}
Consider the game gadget constructed in Figure~\ref{fig:half-gadget}, in any $\eps$-approximate equilibrium of the threshold game $\g(V, E, \frac{1}{2})$, we have $\bx[v] = \frac{1}{2}\max\{\bx[v_1], \frac{1}{2}\}\pm 5\eps$. 
\end{lemma}
\begin{proof}
    By Lemma~\ref{lem:element-gadget}, we have $\bx[v_a] = \max\{\frac{1}{2} - \bx[v_1], 0\} \pm \eps = \frac{1}{2} - \bx[v_1] \pm 2\eps$. 
    The second equality holds since we assume $\bx[v_1] \leq \frac{1}{2} + \eps$ . 
    Consequently, we have $\bx[v_b] = \max\{\frac{1}{2} - (\frac{1}{2} - \bx[v_1]) - \bx[v], 0\} \pm 3\eps = \max\{ \bx[v_1] - \bx[v], 0\} \pm 3\eps$. 
    We consider two cases.
    
    (1) If $\bx[v_1] \geq \bx[v]$, then we have $\bx[v_b] = \bx[v_1] - \bx[v] \pm 3\eps$ and $\bx[v_c] = \max\{\frac{1}{2} - \bx[v_1] + \bx[v], 0\} \pm 4\eps = \frac{1}{2} - \bx[v_1] + \bx[v]\pm 4\eps$. 
    Now consider the fourth gadget, we conclude that $\bx[v] = \max\{\frac{1}{2} - (\frac{1}{2} - \bx[v_1] + \bx[v]), 0\} \pm 5\eps = \bx[v_1] - \bx[v] \pm 5\eps$, which implies that $\bx[v] = \frac{1}{2}\bx[v_1] + 3\eps$.
    
    (2) If $\bx[v_1] < \bx[v]$, then we have $\bx[v_b] = 0 \pm 3\eps$. This implies $\bx[v_c] = \frac{1}{2} \pm 4\eps$ and $\bx[v] =0 \pm 5\eps$. 
    Since we assume $\bx[v_1] < \bx[v]$, it follows that $\bx[v_1] < 5\eps$ and $\bx[v] = \frac{1}{2}\bx[v_1]\pm 5\eps$. Hence, we conclude the proof. 
\end{proof}

Logic gates can be implemented simply as $G_{+}, G_{-}$ and $G_{\frac{1}{2}}$. However, gates constructed this way are not error resilient, i.e., they would amplify the input errors, and this would only establish PPAD-hardness for approximation $\eps = 1/\poly(n)$.  Error resilience brings this up to a fixed constant.
The idea is that we first convert $\{0, \frac{1}{2}\}$ to $\{0, 1\}$,  do logic operations on $\{0, 1\}$ (which is error resilient), and finally transform back to $\{0, \frac{1}{2}\}$

We first construct a gadget $G_{1, \frac{1}{2}}(\nm{v_1} v)$ that transforms $\{0, 1\}$ to $\{0, \frac{1}{2}\}$. We can simply use the COPY gate and it satisfies
\begin{align*}
    \bx[v] = \left\{
    \begin{matrix}
    0 \pm 3\eps & \bx[v_1] = 0 \pm \eps\\
    \frac{1}{2} \pm 3\eps & \bx[v_1] = 1\pm \eps
    \end{matrix}\right..
\end{align*}
\begin{figure}[!h]
		\centering
		\begin{minipage}[t]{.3\textwidth}
			\centering
			\begin{tikzpicture}
			\node (0) at (-1, 0.5) [circle, draw, very thick] {};
			\node (1) at (0, 0) [circle, draw, very thick] {};
			\node (2) at (-1, -0.5) [circle, draw, very thick] {};
			
			\draw[->,  very thick] (0) -- (1);
			\draw[->,  very thick] (2) -- (1);

			\node[text width=1cm] at (-1.2, 0.5) {$x_1$};
			\node[text width=1cm] at (-1.7, -0.5) {$\frac{1}{2}-x_2$};

			\end{tikzpicture}
			\caption{Comparison gadget.}
			\label{pic:comp-gadgets}
		\end{minipage}%
		\begin{minipage}[t]{.3\textwidth}
			\centering
			\begin{tikzpicture}
			\node (0) at (-1, 0.5) [circle, draw, very thick] {};
			\node (1) at (0, 0) [circle, draw, very thick] {};
			\node (2) at (-1, -0.5) [circle, draw, very thick] {};
			\node (3) at (1, 0) [circle, draw, very thick] {};
			
			\draw[->,  very thick] (0) -- (1);
			\draw[->,  very thick] (2) -- (1);
			\draw[->,  very thick] (1) -- (3);

			\node[text width=1cm] at (-1.2, 0.5) {$x_1$};
			\node[text width=1cm] at (-1.2, -0.5) {$x_2$};
			\node[text width=2cm] at (1.4, 0.3) {$x_1\vee x_2$};
			\end{tikzpicture}
			\caption{OR gadget.}
			\label{pic:or-gadgets}
		\end{minipage}
				\begin{minipage}[t]{.3\textwidth}
			\centering
			\begin{tikzpicture}
			\node (0) at (-0.4, 0) [circle, draw, very thick] {};
			\node (1) at (0.6, 0) [circle, draw, very thick] {};
			
			\draw[->,  very thick] (0) -- (1);
			
			\node[text width=1cm] at (0, 0.4) {$x$};
			\node[text width=1cm] at (1, 0.4) {$\bar{x}$};

			\node[text width=0.01cm] at (-1,0.5) {};
			\node[text width=0.01cm] at (-1,-0.5) {};
			\end{tikzpicture}
			\caption{Not gadget.}
			\label{pic:not-gadgets}
		\end{minipage}
	\end{figure}
	
Before we construct $G_{\frac{1}{2}, 1}(\nm{v_1} v)$, which transforms $\{0, \frac{1}{2}\}$ to $\{0, 1\}$, we construct the comparison gadget.

(6) COMPARE $G_{<}(\nm{v_1, v_2} v)$. 
We first apply $G_{\frac{1}{2}-}(\nm{v_2} v_a)$, then connect $v_a$, $v_1$ to a new vertice $v_b$ (see fig.~\ref{pic:comp-gadgets}). 
Finally, we concatenate a gadget $G_{1, \frac{1}{2}}(\nm{v_b} v)$. 
Since we will never compare numbers greater than $\frac{1}{2} + \eps$, we have $\bx[v_a] = \max\{\frac{1}{2} - \bx[v_2], 0\} \pm \eps = \frac{1}{2} - \bx[v_2] \pm 2\eps$ and $\bx[v_a] + \bx[v_1] = \frac{1}{2} - \bx[v_2] + \bx[v_1] \pm 2\eps $. 
If $\bx[v_1] < \bx[v_2] - 3\eps$, then $\bx[v_a] + \bx[v_1] < \frac{1}{2} - \eps$.
This implies $\bx[v_b] = 1 \pm \eps$ and $\bx[v] = \frac{1}{2} \pm 3\eps$. 
If $\bx[v_1] > \bx[v_2] + 3\eps$, then $\bx[v_a] + \bx[v_1] > \frac{1}{2} + \eps$ and $\bx[v_b] = 0 \pm \eps$. 
It then follows that $\bx[v] = 0\pm 3\eps$. 

Now, we can just use a truncated version $G_{<}(\nm{\frac{1}{4}, v_1} v)$, which does not include the last $G_{1, \frac{1}{2}}(\nm{v_b} v)$ gadget), as the transformer gadget $G_{\frac{1}{2}, 1}(\nm{v_1} v)$. 
It satisfies
\begin{align*}
    \bx[v] = \left\{
    \begin{matrix}
    0 \pm \eps & \bx[v_1] = 0 \pm \eps\\
    1 \pm \eps & \bx[v_1] = \frac{1}{2}\pm \eps
    \end{matrix}\right..
\end{align*}

We next construct the logic gates.

(7) OR $G_{\vee}(\nm{v_1, v_2} v)$.
We first apply transformation gadgets $G_{\frac{1}{2}, 1}(\nm{v_1} v_a)$ and $G_{\frac{1}{2}, 1}(\nm{v_2} v_b)$, then connects $v_a$ and $v_b$ to a new vertex $v_c$, which is then connected to the vertex $v_d$ (see Figure~\ref{pic:or-gadgets}).
Finally, we concatenate the transformation gate $G_{1, \frac{1}{2}}(\nm{v_d} v)$. 
Suppose $\bx[v_1] = \frac{1}{2} \pm 3\eps$, it then follows $\bx[v_a] = 1\pm 3\eps$. Therefore, $\bx[v_c] = 0 \pm \eps$ and $\bx[v_d] = 1\pm \eps$. 
After the transformation, we have $\bx[v] = \frac{1}{2} \pm 3\eps$. 
Similarly, $\bx[v_1] = \frac{1}{2} \pm 3\eps$.
On the other hand, if $\bx[v_1], \bx[v_2] = 0 \pm 3\eps$, then we have $\bx[v_a], \bx[v_b] = 0 \pm 6\eps$ and $\bx[v_c] = 1\pm \eps$. 
Hence, $\bx[v_d] = 0 \pm \eps$. After the transformation, we get $\bx[v] = 0 \pm 3\eps$.

(8) NOT $G_{\vee}(\nm{v_1} v)$. 
We first apply $G_{\frac{1}{2}, 1}(\nm{v_1} v_a)$. 
We then connect $v_a$ to a new node $v_b$, and apply $G_{1, \frac{1}{2}}(\nm{v_b} v)$ (see Figure~\ref{pic:not-gadgets}); proof omitted.

We can construct the AND gadget $G_{\wedge}(\nm{v_1, v_2} v)$ using $G_{\vee}(\nm{v_1} v)$ and $G_{\neg}(\nm{v_1, v_2} v)$. 
Thus far, we have constructed all 9 types of gates $\{G_{\xi}, G_{\times \xi}, G_{=}, G_{+}, G_{-}, G_{<}, G_{\wedge}, G_{\vee}, G_{\neg}\}$, and therefore, we conclude the proof for Theorem~\ref{thm:threshold-hard}.
\end{proof}

	\fi

\end{document}